\documentclass{sig-alternate}
\usepackage{graphicx}
\usepackage{amsmath,amssymb}
\usepackage{epsfig}
\usepackage{url}
\usepackage{hyperref}
\usepackage{caption}
\usepackage{adjustbox}
\usepackage{booktabs}     
\usepackage{multirow}
\usepackage{siunitx}	
\usepackage{rotating}
\usepackage{enumerate}
\usepackage{algorithm}
\usepackage{algpseudocode}
\algnewcommand{\LineComment}[1]{\Statex \(\triangleright\) #1}
\algnewcommand{\LineCommentSpaced}[2]{\Statex \(#2 \triangleright\) #1}

\newcolumntype{R}[2]{%
    >{\adjustbox{angle=#1,lap=\width-(#2)}\bgroup}%
    l%
    <{\egroup}%
}
\newcommand*\rot{\multicolumn{1}{R{25}{1em}}}

\newcommand{\spara}[1]{\smallskip\noindent{\bf #1}}

\newtheorem{remark}{Remark}

\newtheorem{theorem}{Theorem}
\newtheorem{corollary}{Corollary}

\newtheorem{lemma}{Lemma}
\newtheorem{problem}{Problem}

\DeclareMathOperator*{\argmin}{arg\,min}
\DeclareMathOperator*{\argmax}{arg\,max}

\newcommand{\NPhard}{$\mathbf{NP}$-hard}

\renewcommand{\algorithmicrequire}{\textbf{Input:}}
\renewcommand{\algorithmicensure}{\textbf{Output:}}

\newcommand{\squishlist}{
 \begin{list}{$\bullet$}
  {  \setlength{\itemsep}{0pt}
     \setlength{\parsep}{3pt}
     \setlength{\topsep}{3pt}
     \setlength{\partopsep}{0pt}
     \setlength{\leftmargin}{2em}
     \setlength{\labelwidth}{1.5em}
     \setlength{\labelsep}{0.5em}
} }
\newcommand{\squishlisttight}{
 \begin{list}{$\bullet$}
  { \setlength{\itemsep}{0pt}
    \setlength{\parsep}{0pt}
    \setlength{\topsep}{0pt}
    \setlength{\partopsep}{0pt}
    \setlength{\leftmargin}{2em}
    \setlength{\labelwidth}{1.5em}
    \setlength{\labelsep}{0.5em}
} }

\newcommand{\squishdesc}{
 \begin{list}{}
  {  \setlength{\itemsep}{0pt}
     \setlength{\parsep}{3pt}
     \setlength{\topsep}{3pt}
     \setlength{\partopsep}{0pt}
     \setlength{\leftmargin}{1em}
     \setlength{\labelwidth}{1.5em}
     \setlength{\labelsep}{0.5em}
} }

\newcommand{\squishend}{
  \end{list}
}

\newcommand{\ourprob}{\textsc{Min Wiener Connector}}
\newcommand{\ourprobweak}{\textsc{Min-A Connector}}
\newcommand{\steiner}{\textsc{Steiner Tree}}
\newcommand{\steinernode}{\textsc{Node-weighted Steiner Tree}}
\newcommand{\setcover}{\textsc{Set Cover}}

\newcommand{\weakrootedprob}{\textsc{Min Weak-A Rooted Connector}}
\newcommand{\weirdprob}{\textsc{Min-B Rooted Steiner Tree}}
\newcommand{\vertexcovershort}{\textsc{Vertex Cover}}
\newcommand{\vertexcover}{\textsc{Vertex Cover in Bounded Degree Graphs}}
\newcommand{\ouralg}{\textsf{WienerSteiner}}

\newcommand{\wsq}{\textsc{ws-q}}

\newcommand{\cep}{\textsc{cps}}
\newcommand{\ctp}{\textsc{ctp}}
\newcommand{\st}{\textsc{st}}
\newcommand{\ppr}{\textsc{ppr}}

\providecommand{\poly}{{\operatorname{poly}}}

\newcommand{\dataset}[1]{\textsf{#1}}

\newcommand{\mycomment}[1]{}
\newcommand{\reals}{{\mathbb R}}
\newcommand{\naturals}{{\mathbb N}}
\DeclareRobustCommand{\calH}[0]{{\mathcal H}}
\DeclareRobustCommand{\calP}[0]{{\mathcal P}}

\newcommand{\W}{\ensuremath{\mathbf{W}}}
\newcommand{\Wsub}{\ensuremath{W}}
\newcommand{\Wone}{\ensuremath{\mathbf{A}}}
\newcommand{\Wonew}{\ensuremath{\mathbf{\widetilde{A}}}}
\newcommand{\Wonesub}{\ensuremath{A}}
\newcommand{\Wonesubw}{\ensuremath{{A}}}
\newcommand{\Wtwo}{\ensuremath{\mathbf{B}}}
\newcommand{\Wtwosub}{\ensuremath{B}}

\newenvironment{prooftext}[1]{\par\noindent{\bf Proof#1.}\quad}{\nopagebreak$\qed$\\}
\newenvironment{myproof}{\begin{prooftext}{}}{\nopagebreak\end{prooftext}}

\pdfoutput=1 
\begin{document}
\conferenceinfo{SIGMOD'15,}{May 31--June 4, 2015, Melbourne, Victoria, Australia.}
\copyrightetc{Copyright \copyright~2015 ACM \the\acmcopyr}
\crdata{978-1-4503-2758-9/15/05\ ...\$15.00.\\
http://dx.doi.org/10.1145/2723372.2749449}
\title{The Minimum Wiener Connector Problem}

\numberofauthors{2}
\author{
\begin{tabular}{cc}
Natali Ruchansky &  Francesco Bonchi \hspace{2mm}  David Garc\'ia-Soriano \\
\affaddr{Computer Science Dept.}& Francesco Gullo \hspace{2mm} Nicolas Kourtellis \\
 \affaddr{Boston University, USA} &  \affaddr{Yahoo Labs, Barcelona}\\
\sf{natalir@bu.edu} & \sf{\{bonchi,davidgs,gullo,kourtell\}@yahoo-inc.com}
\end{tabular}
}

\maketitle
\sloppy

\begin{abstract}
The Wiener index of a graph is the sum of all pairwise shortest-path distances between its vertices.
In this paper we study the novel problem of finding a \emph{minimum Wiener connector}: given a connected graph $G=(V,E)$ and a set $Q\subseteq V$ of query vertices, find a
subgraph of $G$ that connects all query vertices and has minimum Wiener index.

We show that \ourprob\ admits a polynomial-time (albeit impractical) exact algorithm for the special case where the number of query vertices is
bounded.  We show that in general the problem is \NPhard, and has no PTAS unless $\mathbf{P} = \mathbf{NP}$.
Our main contribution is a constant-factor approximation algorithm running in
time~$\widetilde{O}(|Q||E|)$.

A thorough experimentation on a large variety of real-world graphs confirms that our method returns smaller and denser solutions
than other methods, and does so by adding to the query set $Q$ a small number of ``important'' vertices (i.e., vertices with high \emph{centrality}).
\end{abstract}

\section{Introduction}
\label{sec:intro}

\enlargethispage*{\baselineskip}
Suppose we have identified a set of subjects in a terrorist network suspected of organizing an attack. Which other subjects, likely to be
involved, should we keep under control? Similarly, given a set of patients infected with a viral disease, which other people should we monitor?
Given a set of proteins of interest, which other proteins participate in pathways with them?

Each of these questions can be modeled as a graph-query problem: given a graph $G = (V,E)$ and a set of query vertices $Q \subseteq V$,
find a subgraph $H$ of $G$ which \emph{``explains''} the connections existing among the nodes in $Q$, that is to say that $H$ must be connected and contain
all query vertices  in~$Q$. We call this query-dependent subgraph a \emph{connector}.

While there exist many methods for query-dependent subgraph extraction  (discussed later in Section~\ref{sec:related}), the bulk of this literature
aims at finding a ``community'' around the set of query vertices $Q$: the implicit assumption is that the vertices in $Q$  belong to the same community, and a good solution will contain other vertices belonging to the same community of $Q$. When such an assumption is satisfied, these methods return reasonable subgraphs. But when the query vertices belong to different modules of the input graph, these methods tend to return too large a subgraph, often so large as to be meaningless and unusable in real applications.

The goal of this paper is different, as we do not aim at reconstructing a community. Instead we seek a \emph{small} connector: a connected subgraph of the input graph which contains $Q$ and a small set of \emph{important additional vertices}. These additional vertices could explain the relation among the vertices in $Q$, or could participate in some function by acting as important links among the vertices in~$Q$.  We achieve this by defining a new, \emph{parameter-free} problem where, although the size of the solution connector is left unconstrained, the objective function itself takes care of keeping it small.

Specifically, given a graph   $G = (V,E)$ and a set of query vertices $Q \subseteq V$, our problem asks for the connector $H^*$ minimizing the
sum of shortest-path distances among all pairs of vertices  (i.e., the \emph{Wiener index}~\cite{wiener1947structural}) in the solution~$H^*$:
$$
H^* = \argmin_{G[S] : Q\subseteq S \subseteq V} \sum_{\{u,v\} \in S} d_{G[S]}(u,v)
$$
where $G[S]$ denotes the subgraph induced by a set of nodes~$S$, and $d_{G[S]}(u,v)$ denotes the shortest-path distance between nodes $u$ and $v$ in $G[S]$.
We call $H^*$ the \emph{minimum Wiener connector} for query $Q$.

\enlargethispage*{\baselineskip}
This is a very natural problem to study: shortest paths define fundamental structural properties of graphs,  playing a role in all the basic
mechanisms of networks such as their evolution \cite{kossinets2006empirical} and the formation of communities \cite{girvan02community}. The fraction of shortest paths that a vertex takes part
in is called its \emph{betweenness centrality} \cite{bavelas48}, and is a well established measure of the importance of a vertex, i.e.,  the extent to which an
actor has control over information flow. As our experiments in Section \ref{sec:experiments} show, a consequence of our
definition of minimum Wiener connector is that our solutions tend to include vertices which hold an important position in the network, i.e., vertices with high betweenness centrality.

Consider social and biological networks with their modular structure~\cite{girvan02community} (i.e., the existence of communities of vertices densely connected inside, and sparsely connected with the outside). When the query vertices $Q$ belong to the same community, the additional nodes
added to $Q$ to form the minimum Wiener connector will tend to belong to the same community. In particular, these will typically be vertices with
higher ``centrality'' than those in $Q$: these are likely to be influential vertices playing leadership roles in the community. These might be good
users for spreading information, or to target for a \emph{viral marketing} campaign~\cite{kempe03}.

Instead, when the query vertices in $Q$ belong to different communities, the additional vertices added to $Q$ to form the minimum Wiener connector will
contain vertices adjacent to edges that ``bridge'' the different communities. These also have strategic importance: information has to go over these
bridges to propagate from a community to others, thus the vertices incident to bridges enjoy a strategically favorable position because they can block information, or access it before other individuals in their community.
These vertices are said to span a \emph{``structural hole''}~\cite{Burt92structuralholes}: they are the best candidates to target for blocking the
spread of rumors or viral diseases in a social network, or  the spread of malware in a network of computers. In a protein-protein interaction network these vertices can represent proteins that play a key role in linking modules and whose removal can have different phenotypic effects.

\begin{figure}[t!]
  \vspace{-4mm}
  \centering
   \hspace{-6mm}\includegraphics[width= .52\textwidth]{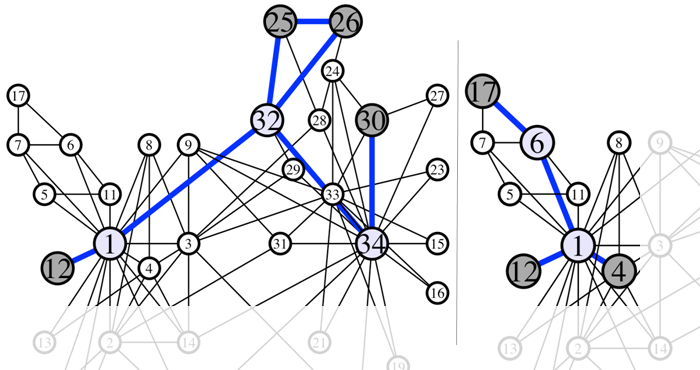}
    \vspace{-6mm}
  \caption{Example of two minimum Wiener connectors on the Zachary's ``karate club'' social network: on the left the query vertices $Q$ (in dark gray)  belong to different communities, on the right they belong to the same community. \label{fig:karate1} }
  \vspace{-3mm}
\end{figure}

As an example, consider the classic  Zachary's ``karate club'' toy social network \cite{karate} with known community structure: a dispute between the club president
(vertex 34) and the instructor (vertex 1) led to the club splitting into two. In Figure~\ref{fig:karate1} we show two different minimum Wiener connectors: the one on the left has the query nodes $Q$ (in dark gray)  belonging to the two different communities, while in the example on the right, all the query vertices belong to the same community.
As discussed above, we can observe that when the query vertices span over different communities, the minimum Wiener connector will include vertices incident to bridging edges. This is the case in our example in Figure~\ref{fig:karate1} (left): given $Q=\{12,25,26,30\}$ the solution subgraph $H^*$ adds to $Q$ the vertices 1 and 34 (the leaders of the two communities) and the vertex 32, which is one of the few vertices connecting 1 and 34 (which do not have a direct connection) and thus practically bridging the two communities.
By contrast, in the example on the right, the query vertices $Q =\{4,12,17\}$ belong to the same community, and as expected the solution
remains inside the community: in this case we just add two vertices, one of which is the community leader (vertex 1), which holds a very central position.


\subsection{Related work}
\label{sec:related}

At a high level our problem can be described as the problem of finding an interesting connected subgraph of $G$ containing a set of query vertices $Q$. Several problems of this type have been studied under different names, depending on the objective function adopted: local community detection, seed set expansion, connectivity subgraphs, just to mention a few.
As discussed above, most of the existing approaches aim at finding a community around the seeds in $Q$: these methods end up producing very large
solutions, especially when the query vertices are not in the same community. Our goal instead is to produce a small connector by adding a few central vertices. Another important distinction is that many researchers have considered only the cases where $|Q| = 1$~\cite{Asur,OverlappingCSSIGMOD13,SozioLocalSIGMOD14,KtrussSIGMOD14} or $|Q| = 2$~\cite{connect}. Finally, our method is \emph{parameter-free}, while several existing methods have many parameters which make a direct comparison complicated.  In the following we provide a brief overview of this body of literature, highlighting the distinctions w.r.t. our proposal.

\spara{Random-walk methods.} Many authors have adopted random-walk-based approaches to the problem of finding vertices related to a seed of vertices: this is the basic idea of Personalized PageRank \cite{PPR1,PPR2}.
Spielman and Teng propose methods that start with
a seed  and sort all other vertices by their degree-normalized
PageRank with respect to the seed \cite{Spielman}. Andersen and Lang \cite{Andersen1} and Andersen et al. \cite{Andersen2} build on these methods to formulate an algorithm for detecting overlapping
communities in networks. In a recent work, Kloumann and Kleinberg \cite{Kloumann} provide a systematic evaluation of different methods for \emph{seed set expansion} on graphs with known community structure. They assume that the seed set $Q$ is made of vertices belonging to the same community $C$: under this assumption they measure precision and recall in reconstructing $C$. Their main findings are that $(i)$ PageRank-based methods outperform other methods, $(ii)$ few iterations (two or three) of the PageRank update rule are sufficient for convergence, and $(iii)$ standard PageRank is to be preferred over degree-normalized PageRank \cite{Andersen1,Andersen2}.

Closer to our goals, Faloutsos et al.~\cite{connect} address the problem of finding a subgraph that
connects two query vertices ($|Q| = 2$) and contains at most $b$ other vertices, optimizing a measure of proximity based on \emph{electrical-current flows}.
Tong and Faloutsos \cite{CenterpieceKDD06} extend the work of \cite{connect} to deal with query sets of any size, but again having a budget $b$ of additional vertices.
They introduce the concept of \emph{Center-piece Subgraph}, the computation of which is based on the Hadamard (i.e., component-wise)  product of a set
of vectors, where each vector is obtained by doing a random walk with restart from a query vertex. The efficiency and scalability of the method is severely limited by the processing time of random walks with restart.
Koren et al.~\cite{KorenTKDD07} redefine proximity using the notion of \emph{cycle-free effective conductance}  and propose a branch and bound algorithm.

All the approaches described above require several parameters: common to all is the size of the required solution, plus all the usual parameters of
PageRank methods, e.g., the jumpback probability, or the number of iterations. We recall that instead our problem definition and algorithms are
completely parameter-free.

\spara{Other methods.} Asur and Parthasarathy \cite{Asur} introduce the concept of \emph{viewpoint
neighborhood analysis} in order to identify neighbors
of interest to a particular source in a dynamically evolving
network. The authors also show a connection of their measure with heat
diffusion. However, the method of  Asur and Parthasarathy has several parameters, such as the budget, the stopping threshold, and minimum number of viewpoint neighborhoods for a vertex.

More recently, Sozio and Gionis~\cite{SozioKDD10} provide a parameter-free  combinatorial optimization formulation. Their problem asks to find  a connected subgraph containing $Q$ and maximizing the minimum degree. Sozio and Gionis show that the problem is solvable in polynomial time and propose an efficient
algorithm. However, their algorithm tends to return extremely large solutions 
(it should be noted that for the same query $Q$ many different optimal solutions of different sizes exist).
To circumnavigate this drawback they also study a constrained version of
their problem, with an upper bound on the size of the output community. In this case the problem becomes \NPhard.
The authors propose a heuristic where the quality of the solution produced (i.e., its minimum degree) can be arbitrarily far away from
the optimal value of a solution to the unconstrained problem.

Cui et al. \cite{SozioLocalSIGMOD14} propose a local-search method to improve the
efficiency of the algorithm by Sozio and Gionis~\cite{SozioKDD10}; however, their method does not solve the issue of the size of the solutions produced. Moreover, their method works only for the special case $|Q| = 1$.


\spara{Steiner Tree and MAD Spanning Trees.}
Given a graph and a set of terminal vertices, the \emph{Steiner tree} problem asks to find a minimum-cost tree that connects all terminals.
This is an extremely well-studied problem: a plethora of methods to solve/approximate it and many variants of the problem have been defined~\cite{Hwang1992}.
We will explain in detail (Section~\ref{sec:problem}) how our {\ourprob} problem differs from the Steiner tree problem.

Another related problem is \emph{Minimum Average Distance}~(MAD)~\emph{Spanning Trees}: given a graph $G$, find a spanning tree of $G$ that minimizes the average shortest-path distance among all pairs of vertices~\cite{Hu1974}.
This problem is related to Wiener index as a MAD Spanning Tree is a spanning tree that minimizes the Wiener index. However, this problem still remains different from our \ourprob\ as the latter asks for subgraphs containing a given set of query vertices rather than asking to span the whole input graph.
In a sense, our problem is to MAD Spanning Trees as Steiner Tree is to Minimum Spanning Tree.

\spara{Wiener index.}
The notion of \emph{Wiener index} is rooted in chemistry, where in 1947 Harry Wiener introduced it to characterize the topology of chemical compounds~\cite{wiener1947structural}.
In general, the Wiener index captures how well connected a set of vertices are, thus bearing resemblance to centrality measures and finding
application in several fields, such as communication theory, facility location, and cryptography~\cite{Dobrynin2001}.
A recent work also considers the Wiener index in the context of event detection in activity networks~\cite{Rozenshtein2014}.
Existing literature on Wiener index focuses on computing it efficiently~\cite{Mohar1988}, finding a tree that minimizes/maximizes it among all trees with a prescribed degree sequence~\cite{Wang2008,Zhang2008,CEla2011}, characterizing  the trees which minimize~\cite{Fischermann2002} or maximize~\cite{Fischermann2002,Stefanovic2008} the Wiener index among all trees of a given size and maximum degree, or solving the inverse Wiener index problem~\cite{Fink2012}.

\emph{To the best of our knowledge, the problem of finding a minimum-Wiener-index subgraph containing a given set of query vertices has never been studied before.}

In our experiments in Section \ref{sec:experiments} we will compare our method with prior contributions which allow $|Q| > 2$: following the findings of \cite{Kloumann} we will use a standard PageRank (with no normalization) personalized over the query vertices $Q$ (\ppr\ for short),  the so-called
\emph{Center-piece Subgraph} \cite{CenterpieceKDD06} (\cep\ for short) which is closer in spirit to our goal of finding a connector and not a community, the so-called \emph{Cocktail Party Problem}~\cite{SozioKDD10} (\ctp) which is parameter free, and the classic Steiner Tree (\st\ for short).

\subsection{Contributions and roadmap}
In this paper we initiate the study of the \ourprob\ problem, a novel \emph{parameter-free} graph query problem,
whose objective function favors small connected subgraphs, obtained by adding few \emph{central} vertices to the query vertices. Beyond this main
contribution, we provide a series of theoretical and empirical results:
\squishlist
    \item We show that, when the number of query vertices is small, \ourprob\ can be solved exactly in polynomial time (\textsection\ref{sec:exactalg}). However, in the general case
    our problem is \NPhard\ and it has no PTAS unless $\mathbf{P}=\mathbf{NP}$ (\textsection\ref{sec:problem}): note that, while the inapproximability result says that the problem cannot be approximated within \emph{every} constant, it leaves open the possibility of approximating it within \emph{some} constant.

    \item In fact, our central result is an efficient constant-factor
    approximation algorithm for \ourprob\ (\textsection\ref{sec:approxalg}), which runs in $\widetilde{O}{(|Q| |E|)}$ time.

    \item We devise integer-programming formulations of our problem (\textsection\ref{sec:lowerbounds}). We use them to compare our solutions for small graphs with those found using
    state-of-the art solvers, and show empirically that our solutions are indeed close to optimal (\textsection\ref{subsec:exp_approx}).

    \item We empirically confirm that existing methods for query-dependent community extraction tend to produce large solutions, which become even larger when the query set $Q$ is made of vertices belonging to different communities (\textsection\ref{subsec:exp_gtc}). Our method instead produces solution subgraphs which are smaller in size, denser, and which include more central nodes  (\textsection\ref{subsec:exp_chara}), regardless of whether the query vertices belong to the same community or not.

    \item We show interesting case-studies in biological and social networks, confirming that our method returns small solutions that include important vertices  (\textsection\ref{sec:casestudies}).

    %

\squishend

\section{Problem statement}
\label{sec:problem}

%
%

\spara{Preliminaries.}
We consider simple, connected, undirected, unweighted graphs. We denote the vertex set (resp., edge set) of a graph $G$ by $V(G)$ (resp., $E(G)$).

    Given a graph $G$ and $S \subseteq V(G)$, we denote by $G[S]$ the subgraph of $G$ induced by $S$: $G[S] = (S, E|S)$, where $E|S = \{(u,v) \in E \mid u \in S, v \in S\}$.
For any connected graph $H$ and $u, v \in V(H)$, let $d_{H}(u,v)$ denote the shortest-path distance between $u$ and $v$ in $H$.
Clearly, if $H$ is a subgraph of $G$, it holds that
$d_G(u,v) \leq d_H(u,v)$.

The \emph{Wiener index} $\W(H)$ of a (sub)graph $H$ is the sum of pairwise distances between vertices in
$H$~\cite{wiener1947structural}:
\begin{equation}\label{eq:wiener}
\W(H) = \sum_{\{u,v\} \subseteq V(H)} d_{H}(u,v),
\end{equation}
where the sum is taken over unordered pairs.

For ease of notation, we identify any $S \subseteq V(G)$ with its induced subgraph $G[S]$.
Thus, we use the shorthand $d_S(u, v)$  (resp., $\W(S)$)
to denote $d_{G[S]}(u, v)$ (resp., $\W(G[S])$.

The input to our problem is a connected graph $G$ and a set $Q \subseteq V(G)$ of \emph{query vertices} (or
        \emph{terminals}). A \emph{connector} for $Q$ in $G$ is a connected subgraph of $G$ containing $Q$.

\spara{Problem definition.}
In this work we aim at finding subgraphs of the input graph that connect a given set of query vertices while minimizing the Wiener index.

\begin{problem}[\ourprob]\label{prob:ourproblem}
Given a graph $G = (V,E)$ and a query set $Q \subseteq V$, find a connector $H^*$ for $Q$ in $G$ with the smallest Wiener index.
\end{problem}
Clearly we may restrict the search to \emph{vertex sets} and their corresponding induced subgraphs.

Note that $\W(H)$ may be written as the product of ${|V(H)| \choose 2}$ and the average distance between  pairs of
distinct vertices of~$H$.
Therefore,  Problem~\ref{prob:ourproblem} encourages solutions that attain a proper tradeoff between having small pairwise distances \emph{and} using few vertices.
In fact, while adding vertices may decrease distances, it also increases the number of terms to be summed up in Eq~\eqref{eq:wiener}.

\spara{Hardness results.} We next prove that the problem is \NPhard, hence unlikely to admit efficient exact solutions. In fact we show a
stronger result, namely that it does not admit a polynomial-time approximation scheme: unless $\mathbf{P} = \mathbf{NP}$, one cannot obtain a
polynomial-time $c$-factor approximation algorithm for every $c > 1$.
This inapproximability result says that the problem cannot be approximated within \emph{every} constant, but leaves open the possibility
of approximating it to \emph{some} constant; in fact we will show this possible.

Our proof makes use of the following inapproximability result for \vertexcover:

\begin{theorem}[Dinur and Safra~\cite{DinurSafra05}]\label{theorem:dinur}
There exist constants $d \in \naturals$, $\alpha \in \reals^+$ with the following property:
given a degree-$d$ graph $G$ and an integer $k \in \naturals$, it is \NPhard\ to distinguish instances where the minimum vertex cover of $G$ has size
larger than $k (1 + \alpha)$ from instances where the minimum vertex cover of $G$ has size at most~$k$.


\end{theorem}


\begin{theorem}\label{theorem:nphard}
There is some constant $\varepsilon \in \reals^+$ such that Problem~\ref{prob:ourproblem} is \NPhard\ to approximate within a $1 + \varepsilon$ factor.
\end{theorem}

\begin{proof}
We present a gap-preserving reduction from \vertexcover\ to \ourprob. Let $\alpha$ and $d$ be as in Theorem~\ref{theorem:dinur}.
Let $\langle G, k \rangle$ be an instance of the decision version of \vertexcover, where the degree of $G$ is at most $d$.
We need to show that there is some constant $\varepsilon > 0$ such that, given $\langle G, k \rangle$, we can construct in polynomial time an instance
$\langle G', Q \rangle$ of
\ourprob\ and a bound $B \in \naturals$ with the following properties:
\begin{enumerate}[(a)]
    \item if $G$ has a vertex cover of size at most $k$, then $\langle G', Q\rangle$ has a connector with Wiener index at most $B$;
    \item if every vertex cover of $G$ has size larger than $k (1 + \alpha)$, then every connector of $\langle G', Q\rangle$ has Wiener index larger than
    $B(1 + \varepsilon)$.
\end{enumerate}

Let $G$ have $n$ vertices and $m$ edges. We may assume that $m \ge k \ge \frac{m}{d}$, for otherwise the vertex cover instance can be answered trivially.
Furthermore, for any fixed constant $c$ we may assume that $k > c\cdot d$ (if not, we can solve the vertex cover
        instance in polynomial time.)

Our graph $G'$ is built as follows.
Let the vertex set of $G'$ be composed of:
\squishlist
    \item a distinguished ``root'' node $r$;
    \item $n$ vertices $v_1, \ldots, v_n$ corresponding to vertices of $G$;
    \item $m$ vertices $e_1, \ldots, e_m$ corresponding to edges of $G$.
\squishend
Put an edge between $r$ and every vertex in $\{v_1, \ldots, v_n\}$, and connect $v_i$ to $e_j$ if and only if $v_i$ is an endpoint of
$e_j$ in $G$. Note that to every $v_i$ correspond at most $d$ such $e_j$'s, and the degree of each $e_j$ is exactly two.
Finally, let $Q = \{e_1, \ldots,  e_m\} \cup \{r\}$ be the set of query vertices.

Observe that any solution to the original \vertexcovershort\ instance gives rise to a feasible solution to \ourprob\  that contains
$Q$ and a subset $\mathcal{X} \subseteq \{v_1, \ldots, v_s\}$; conversely, any solution to \ourprob\ is of the form $Q \cup \mathcal{X}$, where
$\mathcal{X} \subseteq \{v_1, \ldots, v_s\}$ is a vertex subset whose corresponding vertices in $G$
cover all the edges of $G$.

We claim that, if $|\mathcal{X}| = t$ and $Q \cup \mathcal{X}$ forms a connected subgraph, then the Wiener index of the induced subgraph of $G'$ containing $Q$ and $\mathcal{X}$
is at most $u(t) = t^2 + 3 m t + 2m^2$ and at least $l(t) = u(t) - O(m d)$.
Indeed, the contribution to the Wiener index from $r$ to the $t$ chosen vertices is exactly
$t$, and its contribution to the $m$ edges is exactly $2m$. The sum of induced distances among the $t$ chosen vertices is precisely $\binom{t}{2} \cdot 2 = t^2 - t$.
The contribution of the $t$ chosen vertices to $e_1, \ldots, e_n$ is $t (3m - O(d))$, because the distance from $v_i$ to $e_j$ is exactly 3
except in the case that $e_j$ has an edge to $v_i$ (meaning that $e_j$ was an endpoint of $v_i$ in the original graph $G$, and there are at most $d$
        such edges for any $v_i$). Similarly, the contribution from $e_1, \ldots, e_m$ to themselves is $\frac{1}{2} m (4 m - O(d)) = 2m^2 - O(m d)$.
The total is $t^2 + 2m^2 + 3tm - O(m d)$.

If we pick $B = u(k)$, we know that condition (a) is satisfied.
Regarding condition (b), a straightforward computation shows that $u(k) \le 6 d^2 k^2$ and $ u(k(1+\alpha)) \ge u(k) + 5 \alpha k^2,$ hence using the fact
that $m \le k d$ we get
$$ \frac{l((1+\alpha)k) - u(k)}{u(k)} \ge \frac{5 \alpha k^2 - O(m d)}{6 d^2 k^2} \ge \frac{5 \alpha}{6 d^2} - \frac{O(1)}{k}. $$
For some large enough $k_0 = \Theta(d^2/\alpha)$, let $\varepsilon = \frac{5 \alpha}{6 d^2} - \frac{O(1)}{k_0}$; the above shows that $\varepsilon > 0$ and
$$ \frac{l((1 + \alpha) k) - u(k)}{u(k)} > \varepsilon,$$
so condition (b) holds as well, completing the proof.
\end{proof}

\begin{figure}[t]
\centering
\vspace{-4mm}
\includegraphics[width=0.48\textwidth]{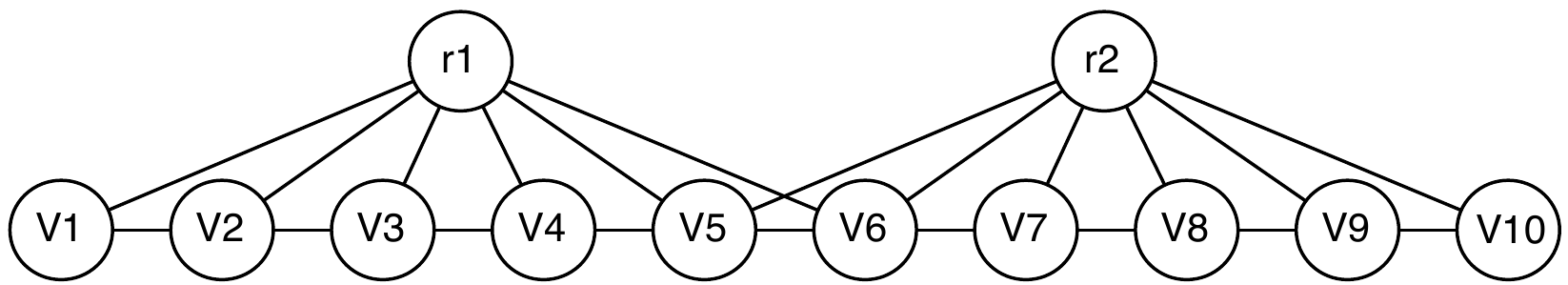}
\caption{\label{fig:steiner_example} Example showing that a solution to \steiner\ may exhibit a large Wiener index. Note that similar situations arise
    also
     in
    real-world instances (see \textsection\ref{subsec:exp_gtc}).\vspace{-3mm}}
\end{figure}

\spara{\ourprob\ vs. Steiner tree.}
At first glance, Problem~\ref{prob:ourproblem} may resemble the well-known \steiner\ problem (see, e.g,~\cite{Hwang1992}): given a graph $G$ and a
set $Q$ of terminal vertices of $G$, find a minimum-sized connector for~$Q$. (The size is measured as
        the total cost of edges used, which for unweighted graphs is one less than the number of vertices used.)
Such an optimal subgraph must be a tree.

Although related (see Section~\ref{sec:approxalg}), the two problems are different. In fact, a solution for \steiner\ may be
arbitrarily bad for \ourprob.
To see this, take a look at the graph in Figure~\ref{fig:steiner_example}.
Let $Q = \{v_1, \ldots, v_{10}\}$ be the set of query/terminal vertices. The unique optimal solution to the \steiner\ problem
is obviously $Q$  itself, which has Wiener index $\W(Q) = 165$.
However, adding either vertex $r_1$ or $r_2$ would lower the index to $\W(Q \cup \{r_1\}) = \W(Q \cup
        \{r_2\}) = 151$, and  the optimal solution to our \ourprob\ problem is given by $Q \cup \{r_1, r_2\}$, which has Wiener index $142$. Also note
that no tree is an optimal solution to this example, showing that the addition of extra edges may help decrease the cost.

In general, the fact that the \steiner\ problem seeks connectors with as few edges/vertices as possible
hinders the minimization of pairwise distances.
We can generalize the example in Figure~\ref{fig:steiner_example} to a line of length $h$ and a root~$r$ connected to all of them. The optimal
\steiner\ solution
exhibits a Wiener index of $\Omega(h^3)$, determined by the $\Omega(h^2)$ pairs of vertices and an average distance of $\Omega(h)$;
on the other hand, a solution to \ourprob\ can include $r$ so as to achieve constant average distance, lowering the Wiener index to ${O}(h^2)$.


\mycomment{
    bullshit!!!!!!1
Similarly, it is also possible to show that, even to obtain approximate solutions, one may need to include many more additional vertices than the
optimal Steiner tree has. Suppose
again that we have $q$ query nodes arranged in a line and a root $r$ connecting to all
query nodes, but this time the path from $r$ to the nodes has length $a = q^{1/4}$. Then the minimum Steiner tree (the line) has $\Theta(q)$ vertices and Wiener index $\Theta(d^3 q^3)$,
      while including the root allows obtaining a connector with $\Theta(q a) = \Theta(d q^{5/4})$ vertices and Wiener index $\Theta((q a)^2 a) =
      \Theta(d^3 q^{11/4})$.
      So both the number of vertices and the Wiener cost are off by a factor of $q^{1/4}$ in the Steiner tree solution.
}

\section{An exact algorithm}\label{sec:exactalg}
\label{sec:algorithms}

Here we address the question of how to solve the \ourprob\ problem exactly.
If the input graph has $n$ vertices, a straightforward solution would be to try all $2^n$ vertex subsets and compute their Wiener index;
this gives a running time of $2^n\,\poly(n)$.
On the other hand, there are some polynomial-time solvable special cases. As an example, for unweighted graphs (like the ones studied in this work), when $|Q| = 2$, any shortest path between the two terminals yields an optimal solution.
As many problems like satisfiability, coloring, etc., turn from easy to \NPhard\ as the ``size'' parameter switches from 2 to 3, it is natural to
wonder if the same happens here. Interestingly, the answer is negative:
the problem admits an exact algorithm that runs in polynomial time for any \emph{fixed} bound on the maximum size of the query set.
The result is of limited practical interest, but gives insight into the nature of the problem.

\begin{theorem}\label{thm:exact_ub}
The \ourprob\ problem can be solved in polynomial time when $|Q| = O(1)$.

That is, there is a function $f : \naturals \to \naturals$ such that the \ourprob\ problem can be solved in time $n^{f(|Q|)}$.
\end{theorem}
The proof is in Appendix~\ref{sec:more_proofs_exact}.
The intuition is that an optimal solution has $f(|Q|)=\poly(|Q|)$ ``pivotal'' vertices that are useful in connecting several query vertices
together or are query vertices themselves. Any other vertices in the solution are simply
``pass-through'' vertices needed to connect pairs of pivotal vertices via shortest paths; they could be replaced by vertices in another arbitrary shortest path between the
required pivotal vertices.
Thus, if we try all possible sets of pivotal vertices and connection patterns among them, and then find shortest paths in $G$ to actually
connect them, we are guaranteed to find an optimal solution.

\section{An approximation algorithm}\label{sec:approxalg}

As we cannot hope for efficient exact solutions to \ourprob,
in this section we design an efficient algorithm with provable approximation guarantees.
Specifically, we achieve a constant-factor approximation in roughly the same time it takes
 to compute shortest-path distances from the terminals to every other vertex in the graph.


%

\spara{Proof outline.} We need to introduce a series of relaxations of \ourprob\ to arrive at a problem
for which it is easier to derive an approximation algorithm. First we show that we can approximate the cost of any solution in terms of the number of
vertices in it and the single-source shortest-path distances to a suitably chosen root vertex $r \in V(G)$. Then we introduce a further relaxation
where
distances are measured according to the original graph; using the techniques developed by~\cite{spanning_shortest} to find light approximate
shortest path trees, we show how to make use of a solution to this relaxation. Then we apply a linearization technique to show that if we knew a certain parameter $\lambda$
controlling the ratio between the size of the optimal solution and the sum of distances to $r$ in the optimal solution, we could reduce our problem to \steinernode.
It turns out that our particular instances of the latter problem admit an $O(1)$-approximation (unlike the general case).
Finally, we explain how to search quickly for the correct values of $r$ and $\lambda$; as an optimization, we also prove that we can further
restrict the search of candidates for $r$. Finally, we combine these arguments to prove the correctness and efficiency of our algorithm.

\spara{Step 1: from \ourprob\ to \ourprobweak.}
First we need the following lemma, whose proof may be found in Appendix~\ref{sec:more_proofs_one_source}.

\begin{lemma}\label{lem:one_source}
For any graph $H$,
$$ \min_{r \in V(H)} \sum_{v \in V(H)} d_{H}(v,r) \le \frac{2\,W(H)}{|V(H)|} \le 2 \min_{r \in V(H)} \sum_{v \in V(H)} d_{H}(v,r). $$
\end{lemma}


Lemma~\ref{lem:one_source} justifies the introduction of the following problem.
Given a subgraph $H$ of $G$ and $r \in V(H)$, let
\begin{eqnarray*}
\Wone(H,r) &=& |V(H)|\cdot\sum_{u \in V(H)}\!d_H(u,r) \\
\Wone(H) &=& \min_{r \in V(H)} \Wone(V(H),r).\label{eq:wienerone_r}
\end{eqnarray*}

\begin{problem}[\ourprobweak]\label{prob:ourprobweak}
Given a graph $G$ and a query set $Q \subseteq V(G)$, find
a connector $H$ for $Q$ in $G$ minimizing~$\Wone(H)$.
\end{problem}
Note that standard Steiner tree problems do \emph{not} minimize $\Wone(H)$, but the number (or total cost) of the edges in $H$.

\begin{corollary}
Any $\alpha$-approximate solution to Problem~\ref{prob:ourprobweak} is a $2\alpha$-approximate solution to \ourprob.
\end{corollary}

\spara{Step 2: from \ourprobweak\ to \weakrootedprob\ via distance adjustments.}
One approach to solve Problem~\ref{prob:ourproblem} is to ``guess'' the correct vertex $r$ and then find a connector $H$ for
$Q$ that
minimizes $\Wone(H, r)$. However,
the objective function depends on the induced
distances of the unknown solution. In order to simplify our task, we now  introduce a ``weak'' relaxation of the above problem where shortest-path
distances are measured in the input graph $G$ instead.

Given a subgraph $H$ of $G$ and a vertex $r \in V(H)$, define

\begin{equation}
\Wonew(H,r) = |V(H)|\cdot\sum_{u \in V(H)}\!d_G(u,r)
\end{equation}


\begin{problem}[\weakrootedprob]\label{prob:weakrootedwiener}
Given graph $G$, root $r \in V(G)$ and query set $Q \subseteq V(G)$, find
a Steiner tree $T$ for $Q$ in $G$ minimizing $\Wonew(T)$.
\end{problem}
Here we insist that the solution be a tree (unlike in Problem~\ref{prob:ourprobweak}, where we allowed non-tree solutions, even though an optimal
        solution may easily seen to be a tree as well). The reason will become apparent shortly.

We are now faced with an additional complication, namely that a good solution to \weakrootedprob\ may not give
a good solution to \ourprobweak.
Hence the need to perform a post-processing step on every candidate solution to ensure that distances in the modified solution resembles distances
in~$G$ as closely as
possible. 

\begin{lemma}\label{lem:add_shortest_paths}
Let $T$ be a subtree of $G$ and $r \in V(T)$. 
There is another subtree $T'$ of $G$ with the following properties:
\begin{enumerate}[(a)]
    \item $V(T') \supseteq V(T)$;
    \item $|V(T')| \le (1 + \sqrt{2}) |V(T)|$;
    \item for all $v \in V(T')$, $d_{T'}(r, v) \le (1 + \sqrt{2})\,d_G(r, v)$.
    \item $\sum_{v \in V(T')} d_{G}(r, v) \le \sqrt{2}\, \sum_{v \in V(T)} d_G(r, v)$.
\end{enumerate}
Furthermore, given $T$, a BFS tree from $r$ in $G$, and $d_G(r, v)$ for all $v \in V(G)$,
    it is possible to construct $T'$ in time $O(|V(T)|)$.
\end{lemma}
This follows from a slight modification of an algorithm by Khuller et al for balancing spanning trees and shortest-path
trees~\cite[Lemma~3.2]{spanning_shortest}; although
they state it for \emph{minimum} spanning trees and shortest path trees with the same vertex set, a careful examination of their proof establishes
Lemma~\ref{lem:add_shortest_paths} as well. For completeness, we reproduce the proof in Appendix~\ref{sec:proof_tree}.

\begin{corollary}\label{coro:add_shortest_paths}
Any $\alpha$-approximation to Problem~\ref{prob:weakrootedwiener} can be used to obtain a $(4+3\sqrt{2}) \alpha$-approximation to Problem~\ref{prob:ourprobweak}.
\end{corollary}
\begin{proof}
We can try all possible choices of $r$. For each of them, let $T$ be an $\alpha$-approximation to
Problem~\ref{prob:weakrootedwiener}.
Then we can
find a tree~$T'$ as in Lemma~\ref{lem:add_shortest_paths}. Since $V(T) \subseteq V(T')$, $T'$ is also a connector for $Q$ and satisfies
\begin{align*}
\Wonew(T', r) &= |V(T')| \sum_{v \in V(T')} d_{G}(r, v) \\
             &\le (1 + \sqrt{2}) \,  |V(T)| \sum_{v \in V(T')} d_{G}(r, v) \\
             &\le (1 + \sqrt{2}) \, \sqrt{2} \, |V(T)| \sum_{v \in V(T)} d_{G}(r, v) \\
             &= (2+\sqrt 2)\, \Wonew(T, r),
\end{align*}              
and $\Wone(T', r) \le (1 + \sqrt{2}) \Wonew(T', r) \le (4 + 3\sqrt 2) \Wonew(T, r).$
              \end{proof}

\spara{Step 3: from \weakrootedprob\ to \weirdprob.}
    We
further relax Problem~\ref{prob:weakrootedwiener} so as to employ a modified objective function where
the product between the number of vertices in $H$ and the sum of original distances to the chosen root $r$ is replaced with a linear combination of
the two.
The rationale here is to make the overall objective function linear and, as such, more amenable to standard approximation techniques.

Given (a subgraph induced by) a subset of vertices $H \subseteq V(G)$, a root vertex $r \in V(H)$, and a parameter $\lambda \in \reals^+$, the modified objective we consider is:
\begin{equation}\label{eq:weirdfunction}
\Wtwo(H,r,\lambda) = \lambda~|H| + \frac{\sum_{u \in V(H)} d_G(r,u)}{\lambda}.
\end{equation}

\begin{problem}[\weirdprob]\label{prob:weird}
Given a graph $G$, a query set $Q \subseteq V(G)$, a root vertex $r \in V$, and a parameter $\lambda \in \reals^+$, find a Steiner tree for $Q \cup \{r\}$
in $G$ minimizing $\Wtwo(H, r, \lambda)$.
\end{problem}

We next show that, by choosing $\lambda$ in the proper way, any approximate solution to Problem~\ref{prob:weird} yields an
approximate solution to Problem~\ref{prob:weakrootedwiener} too.
The right choice of $\lambda$ is given by the following lemma,
proved in Appendix~\ref{sec:more_proofs_alpha}.
\begin{lemma}\label{coro:alpha_approx}
For any graph $G$ with $|V(G)| \ge 2$, query set $Q\subseteq V(G)$ and $r \in V(G)$, there is $\lambda \in [1/\sqrt{2}, \sqrt{|V(G)|}]$
such that for any $\alpha \in \reals^+$,
every $\alpha$-approximate solution to Problem~\ref{prob:weird} is also
 an $\alpha^2$-approximate solution to Problem~\ref{prob:weakrootedwiener}.
\end{lemma}

\spara{Step 4: approximating \weirdprob.}
     Our next step aims to find approximate solutions to Problem~\ref{prob:weird}.
To this end, we note that Problem~\ref{prob:weird} can be cast as a \steinernode\ problem, where the cost of a node $u$ is equal to $\lambda + d_G(r, u) / \lambda$.
However, no approximation factor better than $\Omega(\log |Q|)$ is possible in general for \steinernode,  unless
every problem in $\mathbf{NP}$ can be solved in quasipolynomial time~\cite{KleinRavi}.
Nevertheless, we show that our particular problem admits a constant-factor approximation, by shifting the cost from vertices to edges and reducing it
to a classical \steiner\ problem.
The reason is that in our instance the cost of two adjacent vertices from the root $r$ cannot differ by more than~1, thus the overall solution cost is nearly preserved despite the cost shift.
We formalize this intuition next.

\begin{lemma}\label{lem:node2edge}
Given a graph $G = (V, E)$, a query set $Q \subseteq V$, a root vertex $r \in V$, and a parameter $\lambda \in \reals^+$, let $G_{r,\lambda}$ be a weighted graph with vertex set $V$, edge set $E$, and weight on each edge $(u,v)$ equal to $w(u,v) = \lambda + \frac{\max \{d_G(r,u), d_G(r,v)\}}{\lambda}$.
Then any Steiner tree $T$ for $Q \cup \{r\}$ satisfies the following:
\begin{equation*}\label{eq:node2edge}
\Wtwo(T, r, \lambda) - \lambda \ \le \ \sum_{(u,v) \in E(T)} w(u,v) \ \le \ 2~(\Wtwo(T, r, \lambda) - \lambda).
\end{equation*}
\end{lemma}
\begin{proof}
Observe that the cost $w(u,v)$ of each edge $(u,v) \in E(T)$ lies in the range $[\lambda + d_G(r,u)/\lambda, \lambda + (d_G(r,u)+1)/\lambda]$, as $u$ and $v$
are adjacent in $T$. Notice that in every edge $(u, v)$ of $T$, either $u$ is the parent of $v$ or $v$ is the parent of $u$. Hence, writing $A =
V(T)\setminus\{r\}$, we can bound
\begin{equation*}
\underbrace{\sum_{u \in A}\!\! \left(\! \lambda \!+\! \frac{d_{G}(r,u)}{\lambda} \right )}_{\Wtwo(T, r, \lambda)-\lambda} \leq \!\!\! \sum_{(u,v) \in E(T)}
\!\!\!\!\!w(u,v) \leq \! \sum_{u \in A} \!\!\left( \!\lambda \!+\! \frac{d_{G}(r,u) \!+\!1 }{\lambda} \!\right ),
\end{equation*}
The result follows by noticing that the right-hand side is at most
\begin{equation*}
\Wtwo(T, r, \lambda)-\lambda + \frac{|V(T)|-1}{\lambda} \ \leq \ 2~(\Wtwo(T, r, \lambda) -\lambda).
\end{equation*}
\end{proof}

    Lemma~\ref{lem:node2edge} entails a reduction from Problem~\ref{prob:weird} to the well-studied \steiner\ tree problem. The best known algorithm for
    the latter is the 1.39-factor approximation  algorithm of
    Byrka et al.~\cite{best_steiner}.
    However, it is based on solving a linear program, in contrast to quicker combinatorial algorithms that achieve a factor-2
    approximation. The fastest among the latter is due to Mehlhorn~\cite{Mehlhorn88}.

    \begin{corollary}\label{coro:sol_weird}
    A 4-approximation to Problem~\ref{prob:weird} can be computed in time $O(|E| + |V| \log |V|)$,
    provided that shortest-path distances from $Q$ in $G$ have been precomputed.
    \end{corollary}
    \begin{proof}
    We can construct the graph of Lemma~\ref{lem:node2edge} in time $O(|V| + |E|)$, and use the 2-approximation algorithm of~\cite{Mehlhorn88} for Steiner
    tree, which runs in time $O(|E| + |V| \log |V|)$. By Lemma~\ref{lem:node2edge}, the result is a 4-approximation for Problem~\ref{prob:weird}.
    \end{proof}

\spara{Step 5: choosing $r$ and $\lambda$.}
At this point, we know that, with the right choice of $\lambda$ (which depends on the problem instance), we can get a constant-factor approximate
solution to Problem~\ref{prob:weakrootedwiener}.
For any given graph $G$ and query set $Q \subseteq V(G)$, the algorithm would run as follows:
\squishlist
\item For every vertex $r \in V$ do:
\squishlist
\item
Compute $d_G(r, u)$ from $r$ to every other vertex $u$;
\item
Guess $\lambda$ matching the value stated in Lemma~\ref{coro:alpha_approx};
\item
Construct the weighted graph $G_{r, \lambda}$ of Lemma~\ref{lem:node2edge};
\item
Find an $\alpha$-approximate solution $S^*_r$ to the \steiner\ problem on graph $G_{r,\lambda}$ and terminals $Q \cup \{r\}$;
\squishend
\item
Take the $S^*_r$ that minimizes $\Wtwo(S^*_r, r, \lambda)$.
\squishend

However, we still need to explain how to guess $\lambda$. Since there are only $\poly(|V(G)|)$ many possible values for $\lambda^2$, we could try all
of them in polynomial time. A
faster way is to
fix some $\beta > 0$ and then
try all powers of $(1+\beta)$ in the interval $[\sqrt{1/2},\sqrt{|V|}]$, of which there are only $O(\log |V| / \beta)$ many; this will guarantee
that one of the candidate values of~$\lambda$ tried will be off by a factor of at most $1+\beta$. It is not hard to generalize
Lemma~\ref{coro:alpha_approx} to show that using a $1+\beta$ approximation for the true value of $\lambda$ results
in the loss of another multiplicative
$(1+\beta)^2$
factor in the overall approximation. 

\spara{Step 6: restricting the number of root vertices.}
Finally, we show that trying all possible root vertices $r \in V$ is overkill if we are willing to settle for a somewhat larger approximation
factor. The next result shows that we can restrict our search to elements of the query set (notice that an optimal solution to
        Problem~\ref{prob:ourprobweak} is a tree with leaves in $Q$).

\enlargethispage*{2\baselineskip}


\begin{lemma}\label{lem:query_roots}
Let $T$ be a tree, $r \in V(T)$, and let $x^*$ be a leaf of $T$ closest to $r$. Then $\sum_{u \in V(T)} d_T(x^*, u) \leq 3~\sum_{u \in V(T)}
d_T(r, u)$,
    hence $\Wone(T, x^*) \le 3\cdot \Wone(T, r)$.
\end{lemma}

    \begin{proof}
    For any vertex $x \in V(T)$, let $d(x) = \sum_{u \in V(T)} d_T(u, x)$.
    It suffices to show that $d(x^*) - d(r) \leq 2d(r)$.
    To this end, partition $V(T)$ into levels according to the distance to $r$: $L_i = \{ u \in V(T) \mid d_T(r, u) = i \}$.
    Let $\ell = d_T(r, x^*)$ and for $t\in \naturals$ write $L_{\leq t} = \bigcup_{j \leq t} L_j$, $L_{> t} = \bigcup_{j > t} L_j$.
    On the one hand,
    \begin{eqnarray}\label{dx_ub}
    \lefteqn{d(x^*) - d(r) \ = \sum_{u \in V(T)} (d_T(u, x^*) - d_T(u, r))} \\
                &\leq & \sum_{u \in V(T)} | d_T(u, x^*) - d_T(u, r) |
           \ \leq \ \ (|L_{\le \ell}| + |L_{> \ell}|)~\ell\nonumber.
    \end{eqnarray}
    On the other hand, observe that by our choice of $x^*$,
    it is guaranteed that $|L_0| \leq |L_1| \leq \ldots \leq |L_{\ell}|$, as every vertex at level $i < \ell$ has at least one child (and they are distinct as $H$ is acyclic).
    This implies that we can partition $L_{\leq \ell}$ into a collection of pairs $\{a, b\}$ where $a \neq b$ and $d_T(r, a) + d_T(r, b) \geq \ell$, possibly along
    with a singleton element from $L_{\geq \ell / 2}$.
    Therefore, the average distance from the elements of $L_{\leq \ell}$ to $r$
    is at least $\ell / 2$. Furthermore, every element of $L_{> \ell}$ is at distance $> \ell$ from $r$ by definition.
    Hence
    \begin{equation}\label{dr_lb}
    d(r) \geq |L_{\le \ell}| \frac{\ell}{2} + |L_{> \ell}| (\ell + 1).
    \end{equation}
    Combining Equations~\eqref{dx_ub} and~\eqref{dr_lb} yields the result.
    \end{proof}

\spara{Putting it all together.}\enlargethispage*{\baselineskip}
The pseudocode for our approach is shown as Algorithm~\ref{alg:approx}.
The following theorem summarizes the results about solution quality and running time.

\begin{algorithm}[t]
\caption{\ouralg}\label{alg:approx}
\small
\begin{algorithmic}[1]
\Require A graph $G=(V,E)$; a set of query vertices $Q \subseteq V$.
\Ensure A set of vertices $Q \subseteq H^* \subseteq V$.\vspace{2mm}

\State For all $q \in Q$ and for all $u \in V$, compute $d_G(q,u)$
\State $\mathcal{H} = \emptyset$ \Comment set of candidate solutions
\State $\beta \gets \text{any constant} > 0$ \Comment e.g., $\beta = 1$
\For{$t = 1, \ldots, \lceil \log_{1+\beta} |V| \rceil$}
    \State $\lambda \gets (1+\beta)^t$ \Comment guess the right balance

    \For {$r \in Q$}         \Comment guess a ``root'' vertex
        \LineCommentSpaced{Compute $G_{r,\lambda} = (V,E,w)$ (Lemma~\ref{lem:node2edge})}{\qquad\quad}
        \For {$(u, v) \in E$}
            \State $w(u,v) \gets \lambda + \frac{ \max\{d_G(r, u), d_G(r, v)\} }{\lambda}$
        \EndFor
        \State \ $T \gets \Call{ApproxSteinerTree}{G_{r,\lambda}, Q}$
        \State $H \gets \Call{AdjustDistances}{T}$\Comment see Lemma~\ref{lem:add_shortest_paths}

        \State $\calH \gets \calH \cup \{ (H, r) \}$
    \EndFor
    \EndFor
    \State $H^* \gets \arg\min_{(H, r) \in \mathcal{H}} \Wone(H, r)$ \Comment see Remark~\ref{rem:compute_wiener}
\end{algorithmic}
\end{algorithm}

\begin{theorem}\label{thm:main}
There is a constant-factor approximation algorithm for the \ourprob\ problem running in $O\left (|Q|~(|E| \log |V| + |V| \log^2 |V|) \right )$.
\end{theorem}
\begin{proof}
First we prove correctness. Let $H$ denote an optimal solution to \ourprob.
By Lemma~\ref{lem:one_source}, $\Wone(H) \le 2\cdot \W(H)$, so there is $r \in V(H)$ with $\Wone(H, r) \le 2 \cdot W(H)$.
By Lemma~\ref{lem:query_roots}, there exists $q \in Q$ with $\Wone(H, q) \le 3\cdot \Wone(H, r) \le 6 \cdot \W(H)$;
henceforth
we take $q$ to be the ``root'' vertex in our problems.
Let $K$ denote an optimal solution to Problem~\ref{prob:weakrootedwiener} with root $q$; clearly
$\Wonew(K, q) \le \Wonew(H, q) \le \Wone(H, q)$.
    Let $\lambda \in \reals^+$ be as in Lemma~\ref{coro:alpha_approx}, and let $L \subseteq V(G)$ be an optimal solution to
    Problem~\ref{prob:weird}.
    Then for any connector $X$, the conclusion of Lemma~\ref{coro:alpha_approx} says that if $\Wtwo(X, q, \lambda) \le \alpha \Wtwo(L, q, \lambda)$, then $\Wonew(X, q) \le \alpha^2 \Wonew(K, q)$.

    The main loop is guaranteed to try at some point this choice of $q$ and also a $(1+\beta)$-approximation
 $\lambda'$ for $\lambda$. It is readily seen that, for any $Y$, $\Wtwo(Y, q, \lambda') \le (1 + \beta)
 \Wtwo(Y, q, \lambda)$.
    By Corollary~\ref{coro:sol_weird}, we can find an $4$-factor approximation to Problem~\ref{prob:weird} with $q$ and $\lambda'$; in particular we find
    $X$, $q$ and $\lambda'$ with $\Wtwo(X, q, \lambda') \le 4 \Wtwo(L, q, \lambda') \le 4 (1 + \beta) \Wtwo(L, q, \lambda)$.
    Therefore $\Wonew(X, q) \le 16 (1 + \beta)^2 \Wonew(K, q) \le 96 (1 + \beta)^2 \W(H)$.

    By Corollary~\ref{coro:add_shortest_paths}, line~11 obtains a graph $X'$ with $$\Wone(X', q) \le (4 + 3\sqrt 2)
    \Wonew(X, q) \le 96 (1 + \beta)^2 (4 +
            3 \sqrt2) W(H).$$ Therefore, another application of Lemma~\ref{lem:one_source} tells us that $$\W(X') \le
    \Wone(X', r) \le 792 (1 + \beta)^2
    \W(H) = O(\W(H)),$$ as we
    wished to prove.

As for the running time, computing  the initial shortest-path distances (Line~1) takes $O(|Q|~(|V| + |E|))$ time, while the main loop in Lines~3--14 is repeated $O(|Q|\log|V|)$ times.
Lines~6--10 compute the weighted graph $G_{r,\lambda}$ and find an approximated Steiner tree, thereby solving Problem~\ref{prob:weird}. By
Corollary~\ref{coro:sol_weird}, they run in time $O(|E| + |V|\log|V|)$.
Line~11 adjusts large distances and run in linear time (Lemma~\ref{lem:add_shortest_paths}).
Finally, computing $\Wone(H, r)$ in Line~15 can be done in linear time for each element of $\calH$ (of which there are $O(Q \log |V|)$). In summary, the overall runtime of Algorithm~\ref{alg:approx} is $O\left (|Q|~(|E| \log |V| + |V| \log^2 |V|) \right )$.
\end{proof}

\begin{remark}\label{rem:compute_wiener}\enlargethispage*{\baselineskip}
The last line  of Algorithm~\ref{alg:approx} is intended to return the best solution found. It may be replaced with $H^* \gets
\arg\min_{H \mid (H, r)\in \mathcal{H}} \W(H)$, which can only lead to better solutions. The trouble is that computing $\W(H)$ exactly
may be very costly for large $H$; this poses no difficulty in practice as the sets found are typically small. However, for the worst-case analysis of the running time bounds, it is important to use $\Wone(H, r)$ as a proxy for the actual Wiener index~$\W(H)$.
\end{remark}

\section{Lower bounds}\label{sec:lowerbounds}

\enlargethispage*{\baselineskip}
In this section we design methods to
prove \emph{lower bounds} on the optimal Wiener index.
The idea is to have a way to somehow compare the Wiener index of the solution outputted by our method with the optimum.
As the optimal solution is unknown, we compare against a lower bound on its cost.
While this is
pessimistic approach, proving that our solutions are close to the lower bound allows us to state with certainty that they are close to optimal as well.

To compute the desired lower bound, we show an integer-programming formulation of the \ourprob\ problem.
Let $S$ denote the vertices in a feasible solution, i.e., a connector of $Q$ in $G$.
We set a variable $y_u$ to 1 for each $u \in S$ (in particular $y_u = 1$ for all $u \in Q$), and another variable $p_{st}$ for each
pair $s, t \in V(G) \times V(G)$. In the
intended solution, $p_{st}=1$ iff $y_s = y_t = 1$; we model this by the linear constraint $p_{st} \ge y_s + y_t
- 1$.
Notice that the connectivity requirement is equivalent to being able to route an unit of flow from
$s$ to $t$ whenever $p_{st} = 1$. We add two variables $f_{uv}^{st}$ and $f_{vu}^{st}$ for each edge $\{u, v\}$ in
$G$ and each pair $s, t \in V$; $f_{uv}^{st}$ which will be set to one when a fixed shortest path from $s$ to~$t$
traverses edges $u$ to $v$ in that direction.
 For each $s, t$ and  $v \in V\setminus\{s,t\}$, the flow constraints indicate that the net flow
through $v$ is zero: $\sum_{u \in N(v)} [ f^{st}_{uv} - f^{st}_{vu} ] = 0$, where $N(v)$ are the neighbours of $v$ in $G$. Also, the net flow through $s$
must be~$-p_{st}$
and for $t$ must be $p_{st}$. Since $d_S(s, t) = \sum_{u,v} f^{st}_{uv}$ and the latter sum vanishes when $p_{st} = 0$, $\W(S) = \frac{1}{2} \sum_{u,v}^{s,t} f^{st}_{uv}$.
The complete program is shown next.

\begin{figure}[h!]
\hrule
{
\small
\begin{equation}\label{ip1}
\begin{array}{rrclcl}
\displaystyle \min & \multicolumn{3}{l}{\frac{1}{2} \displaystyle \sum_{u,v,s,t} f^{st}_{uv}} \\
\textrm{s.t.}

& \displaystyle \sum_{u \in N(v)} [ f^{st}_{uv} - f^{st}_{vu} ] &=&
  \begin{cases}
     -p_{st}  &\mbox{if } v=s \\
     p_{st} &\mbox{if } v=t \\
     0       &\mbox{otherwise} \\
  \end{cases}
  & \forall s, t, v \in V \\
& f^{st}_{uv}   &\le& y_u           &\forall \{u, v\} \in E  \\
& p_{st}        &\ge& y_s + y_t - 1 &\forall s, t \in V \\
& y_u           &=&   1             &\forall u \in Q \\
& f^{st}_{uv}, p_{st}         &\ge& 0 \\
& y_u &\in& \{0,1\} \\
\end{array}
\end{equation}
}
\hrule
\end{figure}

\begin{theorem}\label{ip_model}
Program~\eqref{ip1} models the \ourprob\ problem.
\end{theorem}
The proof is reported in Appendix~\ref{sec:more_proofs_ip}.


Program~\eqref{ip1} uses more than $2|E| |V|^2$ variables and more than $|V|^3$ constraints, which can be problematic for large graphs.
A way to reduce the size of the program is to ask for minimization of the pairwise sum of distances in the original graph: this is a safe relaxation as our solutions typically respect the original distances.
Applying this relaxation, the objective function becomes a
linear function of $p_{s,t}$, thus eliminating the need for
separate flow variables for each $s, t$ pair and leading to a program significantly smaller in size.

Let $y_u$ and $p_{st}$ be as before.
The Wiener index of any solution  is at least $\sum_{u,v} d_G(u, v) \cdot p_{uv}$. To express the condition that the variables with $y_u = 1$
form a connected subgraph, we add two variables $x_{uv}$ and $x_{vu}$ for each edge $(u, v)$ of $G$. Pick an arbitrary $q
\in Q$ and any directed spanning tree $T_q$ of $S$ rooted at $q$; the intended solution will have $x_{uv} = 1$ if and only if $v$ is the parent of $u$ in
$T_q$. One constraint is that $x_{uv}+x_{vu} \le y_u$: edge $(u, v)$ can be used only in one direction, and in  order to use it from $u$ to $v$, we
must choose~$u$ as well. Also, for any $u \neq q$, $\sum_{u \in N(v)} x_{uv} = y_v$ (any chosen vertex must have exactly one parent in $T_q$). Finally, we need to make sure that
the edges with $x_{uv}+x_{vu}=1$ form an undirected tree. A tree with $k$ vertices has $k - 1$ edges and no cycles; hence we enforce the constraint
$\sum_{u,v} [ x_{uv} + x_{vu} ] =\sum y_u - 1$ and,  in order to avoid cycles, we add constraints saying that the sum of $x_{uv} + x_{vu}$ for all edges $(u, v)$ in every cycle $C$ of $G$ is at most $|C| - 1$.

\begin{figure}[h!]
\hrule
{
\small
\begin{equation}\label{ip2}
\begin{array}{rrclcl}
\displaystyle \min & \multicolumn{3}{l}{\frac{1}{2} \displaystyle \sum_{s,t} d_G(s, t)\cdot p_{st}} \\
\textrm{s.t.}

& \displaystyle \sum_{u \in N(v)} x_{uv} &=& y_u & \forall v \in V\setminus\{q\} \\
& \displaystyle \sum_{\{u,v\}\in E} [ x_{uv} + x_{vu} ] &=& \sum_u y_u  - 1 &\\
& \displaystyle \sum_{i} [ x_{z_i,z_{i+1}} + x_{z_{i+1},z_i} ] &\le & t -1 & \hspace{-0.6cm} \forall \text{cycle } z_0,\ldots,z_t=z_0\\
& x_{uv}+x_{vu}   &\le& y_u           &\forall \{u, v\} \in E  \\
& p_{st}        &\ge& y_s + y_t - 1 &\forall s, t \in V \\
& y_u           &=&   1             &\forall u \in Q \\
& x_{uv}, p_{st}         &\ge& 0 \\
& y_u &\in& \{0,1\} \\
\end{array}
\end{equation}
}
\hrule
\end{figure}

We reduced the number of variables to a more manageable $O(V^2)$, in exchange for
exponentially many constraints (one per cycle in $G$).  This is not a serious issue because the program above has a \emph{separation
oracle}~\cite{ellipsoid_lp}, and
 commercial solvers support the addition of \emph{lazy constraints}~\cite{gurobi}.


\section{Experiments}
\label{sec:experiments}

In this section we report the results of our empirical analysis. Here we anticipate the main findings:
\squishlist
\item Our approximation algorithm produces solutions which are close to optimal  (\textsection\ref{subsec:exp_approx}).

\item When compared to other concepts of  query-dependent subgraphs extraction such as personalized PageRank  \cite{Kloumann} (\ppr), \emph{Center-piece Subgraph} \cite{CenterpieceKDD06} (\cep), or the \emph{Cocktail Party Subgraph}~\cite{SozioKDD10} (\ctp), the minimum Wiener connector is several orders of magnitude smaller in size, it is much denser, and it includes vertices with higher centrality (\textsection\ref{subsec:exp_chara}).

\item When the query set $Q$ includes vertices belonging to different communities, \ppr, \cep, and \ctp\ return solutions that are 5 to 10 times
larger than the case where the whole of $Q$ belongs to the same community. The minimum Wiener connector is only slightly larger  (\textsection\ref{subsec:exp_gtc}).

\item
Steiner tree produces solutions that are much closer to the minimum Wiener connector than the other methods.
However, in addition to having smaller Wiener index (\textsection\ref{subsec:exp_st}), the Steiner-tree solutions are nearly always less dense, and include vertices with lower centrality.
Also, interestingly, the size of our solutions is comparable to the size of Steiner-tree solutions, despite the fact that Steiner tree explicitly optimizes for solution size.
\squishend

\subsection{Experimental set up}
\spara{Algorithms.}
We compare our algorithm $\wsq$ with several alternative methods described next.
Following the literature on random walks with restart~\cite{tong06fast-RWR,fujiwara12fast-topksearchRW,yu2014reversetopk}, \cep\ is initialized with a restart parameter $c$=0.85, number of iterations $m$=100, and a convergence error threshold $\xi$=$10^{-7}$.
To allow \cep\ to converge to the best possible solution, no budget constraint is given a priori: we greedily add to the solution the highest-score vertex, until we connect the vertices in $Q$.
For the personalized PageRank method, \ppr, we use the same settings as \cep, as well as the same way of selecting which and how many vertices to add to the solution.

For \ctp~\cite{SozioKDD10} we found that the parameter-free version typically returns too large solutions (often with a size comparable to the original graph). In order to limit the size of the solutions returned while keeping it parameter-free, we first execute a BFS from each query vertex until all other vertices in $Q$ are connected, among all these subgraphs we pick the smallest one, and run over it the greedy algorithm of~\cite{SozioKDD10}.
For Steiner tree (\st) we use the approximation algorithm by Mehlhorn~\cite{Mehlhorn88}, which is the same that \wsq\ uses internally to solve the
Steiner tree instances it generates (\textsection\ref{sec:approxalg}).
All algorithms are implemented in \texttt{C++}.

\begin{table}[t!]
\vspace{-2mm}
\caption{
Summary of graphs used.
$\delta$: density, ad: average degree, cc: clustering coefficient, ed: effective diameter.
Datasets with ground truth communities ($*$).
Classical Steiner Tree benchmarks with given query workload ($\#$).
}
\centering
\small
\tabcolsep=0.12cm
\begin{tabular}{lrrrrrrr}
\toprule
\hspace{-4mm}& Dataset			&	$|V|$		&	$|E|$		&	$\delta$	&	ad		&	cc		&	ed	\\
\midrule
\hspace{-4mm}& \dataset{football}	&	115		&	613		&	9.4e-2	&	21.3	&	0.40	&	3.9	\\
\hspace{-4mm}& \dataset{jazz}		&	198		&	2742		&	1.4e-1	&	55.4	&	0.62	&	3.8	\\
\hspace{-4mm}& \dataset{celegans}	&	453		&	2025		&	2.0e-2	&	17.9	&	0.65	&	4.0	\\
\hspace{-4mm}& \dataset{email}	&	1133		&	5452		&	8.5e-3	&	9.62	&	0.22	&	8	\\
\hspace{-4mm}& \dataset{yeast}	&	2224		&	6609		&	2.6e-3	&	5.94	&	0.14	&	11	\\
\hspace{-4mm}& \dataset{oregon}	&	10670	&	22002	&	3.8e-4	&	4.12	&	0.30	&	4.4	\\
\hspace{-4mm}& \dataset{astro}	&	18772	&	198110	&	1.1e-3	&	22.0	&	0.63	&	5	\\
\hspace{-4mm}& \dataset{dblp*}		&	317080	&	1049866	&	2.1e-5	&	6.62	&	0.63	&	8.2	\\
\hspace{-4mm}& \dataset{youtube*}	&	1134890	&	2987624	&	4.6e-6	&	5.27	&	0.08	&	6.5	\\
\hspace{-4mm}& \dataset{wiki}		&	2394385	&	5021410	&	1.8e-6	&	4.19	&	0.22	&	3.9	\\
\hspace{-4mm}& \dataset{livejournal}&	3997962	&	34681189	&	4.3e-6	&	17.3	&	0.28	&	6.5	\\
\hspace{-4mm}& \dataset{twitter}	&	11316811	&	85331846	&	1.3e-6	&	15.1	&	0.09	&	5.9	\\
\hspace{-4mm}& \dataset{dbpedia}	&	18268992	&	172183984&	1.0e-6	&	18.9	&	0.17	&	5.0	\\
\hspace{-4mm}& \dataset{puc$^\#$}		&	64-4096		&	448-24574	&	-	&	-	&	-	&	-	\\
\hspace{-4mm}& \dataset{vienna$^\#$}	&	1991-8755	&	3176-14449	&	-	&	-	&	-	&	-	\\
\bottomrule
\end{tabular}
\label{tab:all-nets}
\vspace{-2mm}
\end{table}

\spara{Datasets and query workloads.} \enlargethispage*{\baselineskip}
We use real-world publicly-available graphs of various types and sizes, spanning different domains: communication over emails and wiki pages, citation and co-authorship networks, road networks, social networks, and web graphs  (Table~\ref{tab:all-nets}).

Small datasets are used for assessing approximation quality (\textsection\ref{subsec:exp_approx}) of our algorithm \wsq w.r.t. the best provable bounds obtained by solving the integer program in \textsection\ref{sec:lowerbounds}.

Medium-large datasets are used for characterizing the solutions produced by the various algorithms described above, in terms of size, density, and centrality (\textsection\ref{subsec:exp_chara}).

In all these datasets, the query workloads are made of random query-sets $Q$, with controlled size and average distance of the query vertices.
Datasets marked with (*) contain ground-truth community structure \cite{YangL15}: these are used to create different workloads with query vertices in $Q$ belonging to the same community or to different communities (\textsection\ref{subsec:exp_gtc}).
As we delve deeper in the comparison between \ourprob\ and Steiner tree (\textsection\ref{subsec:exp_st}), we use benchmarks with predefined query workloads which are used for assessing Steiner tree algorithms.\footnote{\url{http://steinlib.zib.de/}}
These are marked with ($\#$) in Table~\ref{tab:all-nets}: benchmark  \dataset{puc} contains 25 problems on small graphs with $|Q| \in [8, 2048]$,  while benchmark  \dataset{vienna} contains 85 problems with $|Q| \in [50, \approx 5k]$.

For scalability assessment (\textsection\ref{subsec:exp_scala}) we use the larger graphs in  Table~\ref{tab:all-nets}, plus synthetic graphs generated according to the Erd\H{o}s-R\'enyi and Power-Law  models.\footnote{\url{http://snap.stanford.edu/snap/index.html}}


\begin{table}[t!]
\vspace{-2mm}
\caption{Comparison of the Wiener index of $\wsq$'s solution with the lower ($G_L$) and upper ($G_U$) bounds found by Gurobi solver for different datasets and query set sizes.
The cost of the optimal solution is guaranteed to be in $[G_L, G_U]$. $\dagger$Numbers based on the best lower bound the solver could prove \emph{before it ran out of memory}; they give an \emph{upper bound} on the error that is likely to be an overestimate.
\label{tab:lowerbounds}}
\centering
\small
\tabcolsep=0.12cm
\begin{tabular}{crrrrr}
\toprule
Dataset		&	|Q|	&	\wsq		&	$G_{U}$	&	$G_{L}$	&	Error interval\\
\midrule
\multirow{4}{-2mm}{\begin{sideways} \dataset{football} $\;$ \end{sideways}}
& 3		&	40		&	40		&	40		&	0			\\
& 5		&	172		&	172		&	164		&	$[0,4.9\%]$	\\
& 10		&	656		&	598		&	538		&	[9.6\%, 22\%]\\
& 20		&	2352		&	2018		&	$1546^\dagger$		&	$ [16.5\%, 52.2\%^\dagger]$	\\
\midrule
\multirow{4}{-2mm}{\begin{sideways} \dataset{jazz} $\;$ \end{sideways}}
& 3		&	16		&	16		&	16		&	0			\\
& 5		&	44		&	44		&	44		&	0	\\
& 10		&	276		&	276		&	260		&	$[0, 6.2\%]$	\\
& 20		&	1014		&	964		&	936		&	$[5.1\%, 8.4\%]$	\\
\midrule
\multirow{4}{-2mm}{\begin{sideways} \dataset{celegans} $\;$ \end{sideways}}
& 3		&	36		&	36		&	36		&	0			\\
& 5		&	106		&	106		&	106		&	0	\\
& 10		&	330		&	330		&	326		&	$[0, 1.3\%]$	\\
& 20		&	1204		&	1196		&	1192		&	$[0.66\%, 1.1\%]$	\\
\midrule
\multirow{4}{-2mm}{\begin{sideways} \dataset{email} $\;$ \end{sideways}}
& 3		&	58		&	58		&	58		&	0			\\
& 5		&	250		&	250		&	240		&	$[0,4.2\%]$	\\
& 10		&	1352		&	1208		&	$1033^\dagger$		&	$[11.9\%, 30.9\%^\dagger]$	\\
& 20		&	5490		&	5490		&	$4032^\dagger$		&	$[0, 36.2\%^\dagger]$	\\
\bottomrule
\end{tabular}
\vspace{-2mm}
\end{table}

\subsection{Approximation quality}\label{subsec:exp_approx}
Table~\ref{tab:lowerbounds} reports the Wiener index of the solution  produced by \wsq, and how it compares with the best provable bounds obtained
with the integer-programming formulation reported in Program~\eqref{ip2} (\textsection\ref{sec:lowerbounds}) and the state-of-the-art Gurobi solver~\cite{gurobi}.
This comparison was carried out on small graphs as otherwise the number of variables would be too large to even formulate the integer program.
We initialize the solver with our solution so that the solver's upper bound can never be worse by construction.
A match in the solver's upper and lower bounds indicates an optimal solution was found.
When they do not coincide, either there is a gap between the best solution and the lower bound from Program~\eqref{ip2}, or the solver ran out of
memory during the optimization phase (in which case we report the best lower bound found so far).

We also report an error interval obtained by comparing our solution with the solver's best upper and lower bounds.
Observe that, for small query sets (three to five vertices), \wsq\ produces solutions that are optimal or very close to it (with error in the interval $[0,5\%]$).
The worst discrepancy between our and the solver's best solution is $16.5\%$ (\dataset{football} with $|Q| = 20$); and here all we can prove is that our solution is at most $52.2\%$ from optimal.
However, note in this case there is also a significant gap between the solver's own lower and upper bounds, thus $52.2\%$ is likely to be an overestimate. It should also be noted that this query set size is approximately $1/5$ the size of the whole vertex set $V$.

\begin{table}[t!]
\vspace{-2mm}
\centering
\caption{Main characteristics of the solution $H$ returned by different algorithms on 6 datasets, with $|Q| = 10$ and average distance of 4 among the vertices in $Q$. Each experiment is run 5 times and we report averages of
size of the solution $|V[H]|$, density of the solutions  $\delta(H) = |E[H]|/{|V[H]| \choose 2}$,  average betweenness centrality $bc(H)$ of vertices in $H$, and Wiener index $\W(H)$.
\label{tab:quality1}}
\vspace{-2mm}
\scriptsize
   \setlength\tabcolsep{3.5pt}
\begin{tabular}{r|rrrrrr|l}
   \multicolumn{1}{c}{\hspace{-6mm}}  &   \rot{\textsf{email}}& \rot{\textsf{yeast}}& \rot{\textsf{oregon}}& \rot{\textsf{astro}} & \rot{\textsf{dblp}} & \rot{\textsf{youtube}}&  \\ \cline{2-7}

\multirow{5}{*}{\rotatebox{90}{$|V[H]|$}}   & 671 & 819 & 9028  &  12758 & 11804 & 17865 &  \ctp \\
                                                                             & 155 & 188 & 4556  & 1735 & 7349 & 5615 & \cep \\
                                                                             & 137 & 100 & 1846  & 598 & 842& 684 &  \ppr\\
                                                                             & 26 & \textbf{24} & 26  & 26 & 25 & 19 &  \st \\
                                                                             & \textbf{24} & \textbf{24} & \textbf{23}  & \textbf{23} & \textbf{23} & \textbf{17} &  \wsq \\
\cline{2-7}

\multirow{5}{*}{\rotatebox{90}{$\delta(H)$}}   & 0.016 & 0.016 & 0.01  & <0.01 & <0.01 & 0.01&  \ctp \\
                                                                                 & 0.047 & 0.028 & 0.02  &  0.019 & 0.01 & <0.01 & \cep \\
                                                                                 & 0.029 & 0.039 & 0.02  & 0.07  & 0.01 & 0.02 &  \ppr\\
                                                                                 & 0.080 & 0.088 & 0.090  & 0.09 & 0.08 & 0.1 &  \st \\
                                                                                 & \textbf{0.093} & \textbf{0.091} &\textbf{0.106}  &  \textbf{0.13} & \textbf{0.11} & \textbf{0.13} &  \wsq \\
\cline{2-7}

\multirow{5}{*}{\rotatebox{90}{$bc(H)$}}       & <0.01 & <0.01 & <0.01  & <0.01 & <0.01 & <0.01 &  \ctp \\
                                            & 0.03 & 0.02 & <0.01  & <0.01 & <0.01 & <0.01 & \cep \\
                                            & 0.03 & <0.01 & <0.01  & 0.02 & 0.01 & <0.01 &  \ppr\\
                                            & 0.09 & 0.07 & 0.10  & 0.11 & 0.10 & 0.13 &  \st \\
                                            & \textbf{0.11} & \textbf{0.11} & \textbf{0.12}   & \textbf{0.14} & \textbf{0.12} & 0\textbf{.18} &  \wsq \\
\cline{2-7}

 \multirow{5}{*}{\rotatebox{90}{$\W(H)$}}       & $\approx 750k$ & $\approx 2M$ & $\approx 137M$  & $\approx 292M$ & $\approx 400M$ & $\approx 1.5G$ &  \ctp \\
                             & $54\,598$ & 69\,296 & $\approx 50M$  & $\approx 8.3M$ & $\approx 12.6M$ & $\approx 561M$ & \cep \\
                             & $52\,222$ & $15\,838$ & $\approx 7.5M$  & $40\,079$ & $\approx 1.2M$ & $\approx 1.3M$ &  \ppr\\
                             & $1\,200$ & $1\,259$ & $1\,164$  & 1\,318 & 3\,371 & 1\,324 &  \st \\
                             & \textbf{968} & \textbf{931} & \textbf{923}  & \textbf{1\,007} & \textbf{2\,043} & \textbf{956} &  \wsq \\
 \cline{2-7}
\end{tabular}
\end{table}
\begin{figure}[t!]
\vspace{-2mm}
\begin{tabular}{c}
\hspace{-10mm}	\includegraphics[scale=0.95]{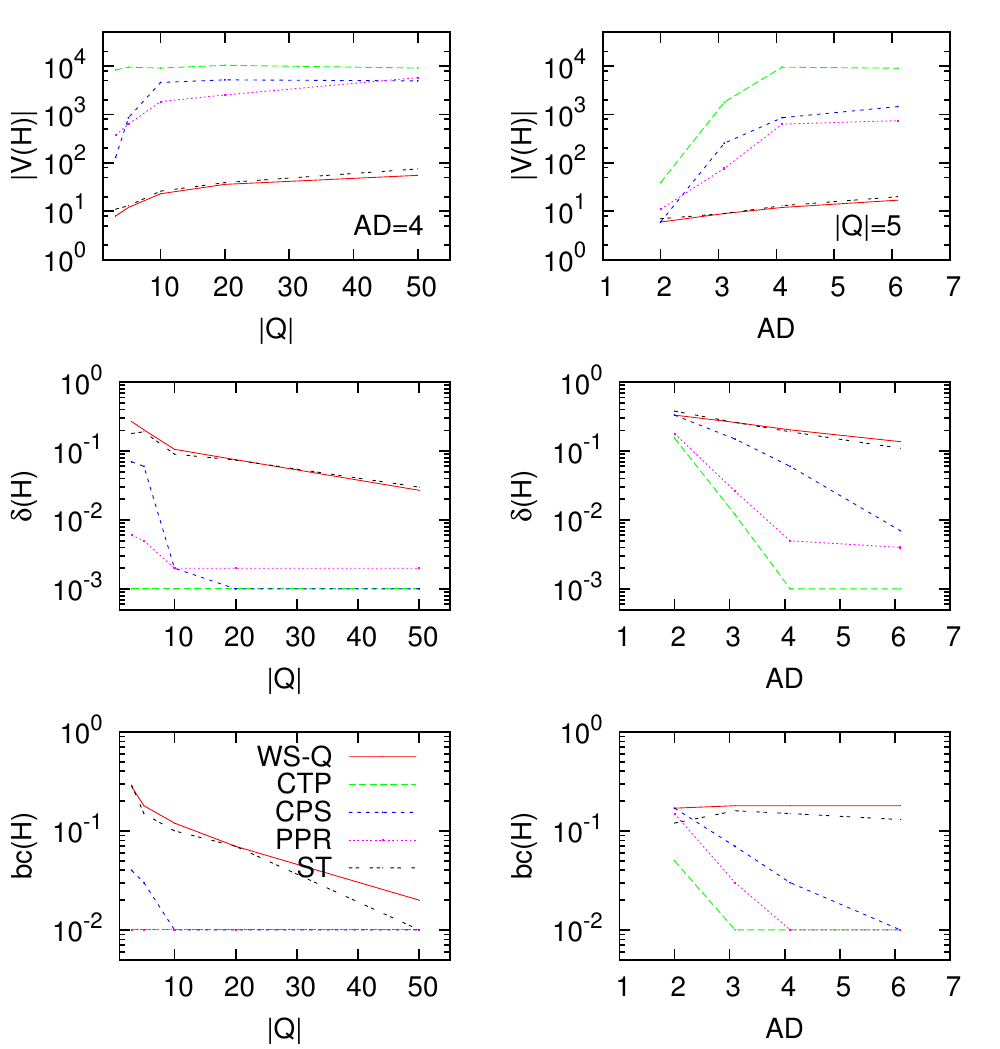}
\end{tabular}
\vspace{-2mm}
\caption{Left column: fixed average distance $AD=4$ among query vertices, varying $|Q|$. Right column: fixed query set size $|Q|=5$, varying average distance among query vertices. We report $|V(H)|$, $\delta(H)$, and  $bc(H)$ on \dataset{oregon}.\label{fig:quality1}}
\vspace{-2mm}
\end{figure}

\subsection{Solution characterization}\label{subsec:exp_chara} \enlargethispage*{\baselineskip}
Table~\ref{tab:quality1} and Figure~\ref{fig:quality1} report a characterization of the solutions produced by the various algorithms in terms of  number of vertices in the solution ($|V(H)|$), density of the solution ($\delta(H)$) and betweenness centrality of the vertices in the solution ($bc(H)$).
Table~\ref{tab:quality1} reports results for various graphs with a fixed size of $Q$, and a fixed average distance among the vertices in $Q$, while Figure~\ref{fig:quality1} shows the same statistics for a single dataset (\dataset{oregon}) with varying size of $Q$ and average distance of the query vertices.
Results confirm that \wsq\ produces solutions which are always smaller, denser and contain vertices with higher betweenness centrality than the other methods. The difference is striking with all the methods, with Steiner tree being much closer to the type of solutions produced by \wsq.
As expected, since the other methods do not try to optimize it, \wsq\ produces solutions with a Wiener index that is orders of magnitude smaller.
Moreover, the solutions \wsq\ provides have much smaller index the Steiner-tree solutions.
A deeper comparison between \wsq\ and \st\ is reported in Section \ref{subsec:exp_st}.

\begin{table}[t!]
\vspace{-2mm}
\centering
\caption{Average solution size for query workloads based on ground-truth communities:   \dataset{dc} = query vertices in different communities, \dataset{sc} = query vertices in the same community, and \dataset{dc/sc} = the ratio of the previous two columns. \label{tab:communities}}
\vspace{-3mm}
\small
   \setlength\tabcolsep{3.5pt}
\begin{tabular}{r|rr|r|rr|r|l}
   \multicolumn{1}{c}{\hspace{-6mm}}  &   \rot{\textsf{dblp-dc}}& \rot{\textsf{dblp-sc}}& \rot{\textsf{dblp:dc/sc}}& \rot{\textsf{youtube-dc}} & \rot{\textsf{youtube-sc}} & \rot{\textsf{youtube:dc/sc}}&  \\ \cline{2-7}

\multirow{5}{*}{\rotatebox{90}{$|V[H]|$}}   & 1.4e5  & 2.8e4 &  \textbf{5.03} &  8e5 &2.3e5    &  \textbf{3.5}&  \ctp \\
                                                                             &  4.1e4   &3.69e3 &\textbf{11.3}  &  3.6e5& 5.0e4 &\textbf{7.4}  & \cep \\
                                                                             & 3.4e4  & 3.5e3 &  \textbf{8.6}  &3.9e5  &4.1e4 & \textbf{9.2} &  \ppr\\
                                                                             &  40 &29  &  \textbf{1.43} &20  & 16 &  \textbf{1.3}&   \st \\
                                                                             & 36 &26 & \textbf{1.38}  & 18 & 14 &\textbf{ 1.3 }&  \wsq \\
\cline{2-7}
\end{tabular}
\vspace{-1mm}
\end{table}

\subsection{Ground-truth communities workload}\label{subsec:exp_gtc}
Next, we compare the behavior of the various methods when the query set $Q$ belongs to a community or to multiple communities. To this end, we use graphs with ground-truth community structure (\dataset{dblp} and \dataset{youtube}) and produce two query workloads for each graph: one with  query vertices belonging the \emph{same community} (denoted \dataset{sc}) and one with query vertices coming from \emph{different communities} (denoted \dataset{dc}). Each workload contains 40 queries, 10 for each size $|Q| \in \{3,5,10,20\}$. For  \dataset{sc} workloads, we pick the community at random, but avoiding small communities (of size smaller than 100 vertices).

The results are reported in Table \ref{tab:communities}. We observe that when $Q$ belongs to multiple communities, random-walk-based methods (\ppr\ and \cep) produce solutions which are from 7 to 11 times larger than when $Q$ belongs to only one community. While the ratio is less striking for \ctp\ (3 to 5 times larger), the solutions produced are already extremely large in both workloads.

The results confirm that these methods are conceived to reconstruct a community around a given seed set of vertices~$Q$, implicitly assumed to belong
to the same community; thus, they tend to return significantly larger results when this does not hold. By contrast, \wsq\ and \st\ do not rely on such
assumptions, and the difference in average solution size between the two workloads is much smaller.

\subsection{Comparison on Steiner tree benchmarks}\label{subsec:exp_st}\enlargethispage*{\baselineskip}
We have shown that the types of solutions produced by community-oriented methods have very different characteristics from those by the minimum Wiener connector and the Steiner tree. We next delve deeper in the comparison between the minimum Wiener connector and the Steiner tree, using
Steiner-tree benchmarks with predefined query workloads, and focusing on the two objective functions of the two problems: the size of the solution (Steiner) and the sum of the pairwise shortest-path distances (Wiener).

Figure~\ref{fig:st-wsq-quality} reports the cumulative distribution functions (CDFs) of the ratio of solution size (left) and Wiener index (right) between \st\ and \wsq, on the two benchmarks  \dataset{puc} and \dataset{vienna}.
As expected, \wsq\ produces solutions which have a much smaller sum of the pairwise shortest-path distances.
Also interesting to observe is the fact that our algorithm \wsq\ often outperforms the well-established Steiner-tree approximation-algorithm by Mehlhorn~\cite{Mehlhorn88} with respect to the size of the solution, which is the objective function of the Steiner tree (recall that the \ourprob\ objective function implicitly favors small solutions).

\begin{figure}[t!]
\centering
\vspace{-2mm}
\begin{tabular}{c}
	\hspace{-6mm}
	\includegraphics[scale=0.9]{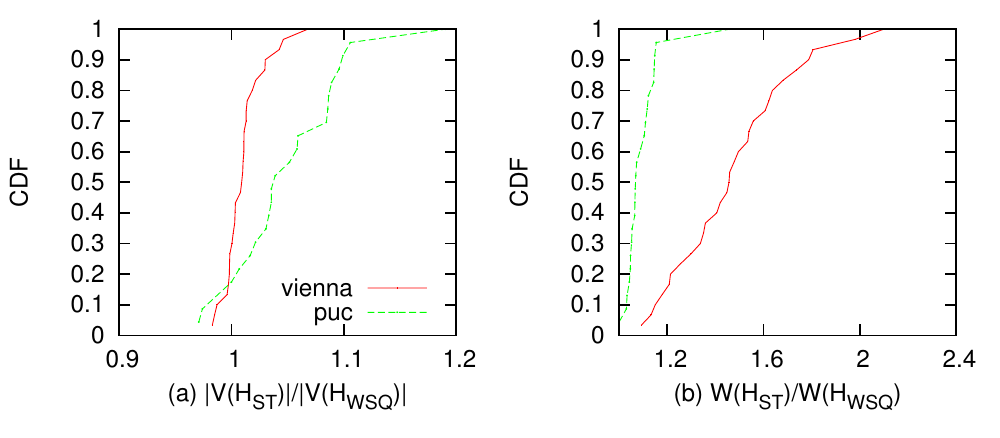}
\end{tabular}
\vspace{-3mm}
\caption{CDFs of the ratio of (a) solution size and (b) Wiener index on benchmarks \dataset{vienna} and \dataset{puc}.\label{fig:st-wsq-quality}}
\vspace{-4mm}
\end{figure}

\begin{figure*}[t!]
\vspace{-4mm}
\centering
\begin{tabular}{cc}
\hspace{-6mm}\includegraphics[scale=1]{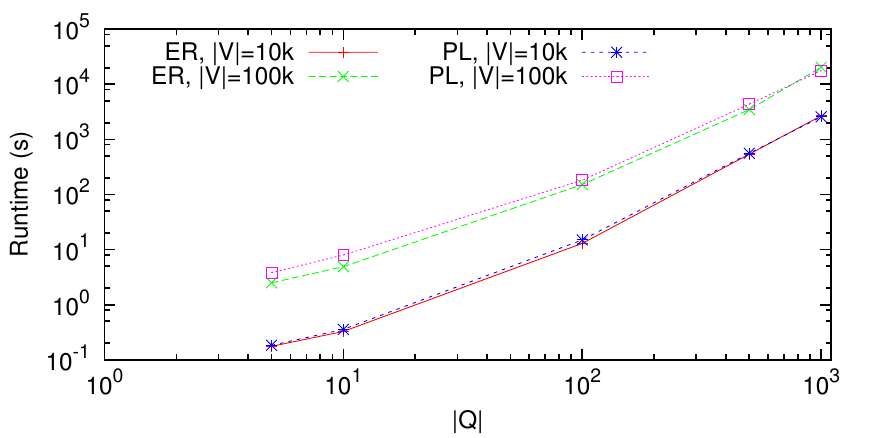} & \includegraphics[scale=1]{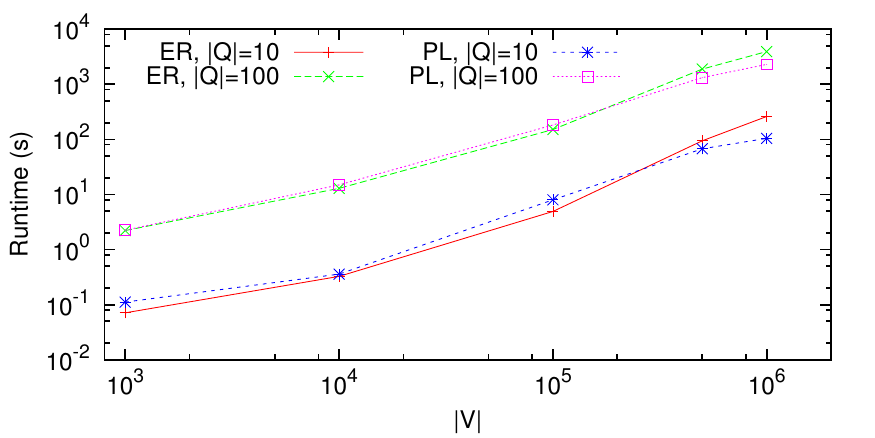}\\
\hspace{-6mm} \includegraphics[scale=1]{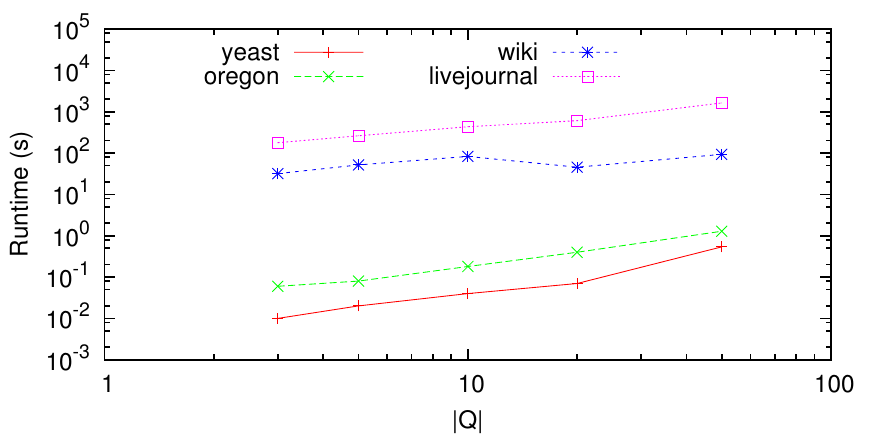}  & \includegraphics[scale=1]{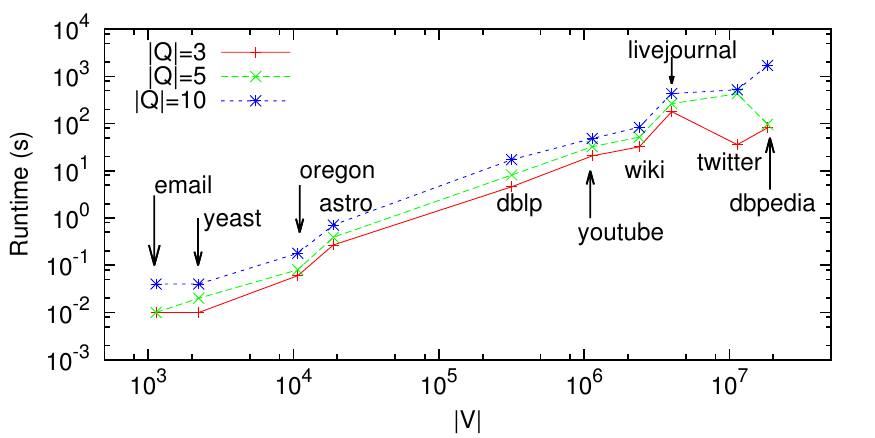}
  \end{tabular}
\vspace{-4mm}
\caption{Computational runtime of \wsq\ on different synthetic Power-Law ($PL$) and Erd\H{o}s-R\'enyi ($ER$) graphs (first row) and real-world graphs (second row),  with varying query set size and fixed graph size (left column), and varying graph size and fixed query set size (right column). \label{fig:runtime}}
\vspace{-3mm}
\end{figure*}

Interestingly, we also observe that many problem instances on the \dataset{vienna} benchmark are real-world instances of the situation depicted earlier in Figure~\ref{fig:steiner_example} -- that is to say \wsq\ produces a slightly larger solution in number of vertices,  yet with a significantly smaller Wiener index.

\subsection{Scalability}\label{subsec:exp_scala} \enlargethispage*{2\baselineskip}
We now focus on \wsq's runtime performance and scalability with an increasing graph or query set size.
We use Erd\H{o}s-R\'enyi ($ER$) and Power-Law ($PL$) models to generate synthetic graphs with varying size, while keeping constant other graph properties. We also use the larger real-world graphs in Table \ref{tab:all-nets}.
Results are reported in Figure~\ref{fig:runtime}.
First, we note that the performance of the algorithm is not significantly affected by the type of graph (random or power-law). Second, runtime has an almost linear relationship with the query set size, as well as the input graph size.
However, as expected by Theorem~\ref{thm:main},  runtime is most impacted by the graph size rather than the query-set size.

\spara{Parallelization.}  Examining Algorithm~\ref{alg:approx} we notice that we can easily speed-up our method
via parallelization (e.g., via a Map-Reduce execution), assuming the graph $G$ fits in memory.
In fact, by launching $|Q|$ threads in parallel (Map), we can achieve a linear speedup of $|Q|$. Each thread examines one root $r \in Q$ and computes
shortest-path distances in~$G$ from $r$ to construct and solve the Steiner tree instances for different choices of the parameter $\lambda$.
Then, all possible solutions can be collected (Reduce), and the best one chosen.
To this end, each thread needs to evaluate its candidate solutions. Since these are typically small in practice, the thread
can compute the induced shortest-path distances from all vertices in its solutions and compute their Wiener indices exactly. In the (unlikely)
scenario that a solution is large, we can instead approximate the Wiener
index (see Remark~\ref{rem:compute_wiener}). This preserves the approximation guarantee while providing an overall speedup of $|Q|$.
If $G$ is too large to fit in memory, it becomes necessary to employ techniques for parallel and/or approximate shortest-distance computations~\cite{distance_oracles,bader06parallel-betweenness,kourtellis12kpath,riondato2014approxbetweenness}, but these are beyond the scope of this work.


\section{Case studies}
\label{sec:casestudies}
%
%
%
%
%

\begin{figure}[t!]
\vspace{-1mm}
\begin{center}
	\includegraphics[scale=0.6]{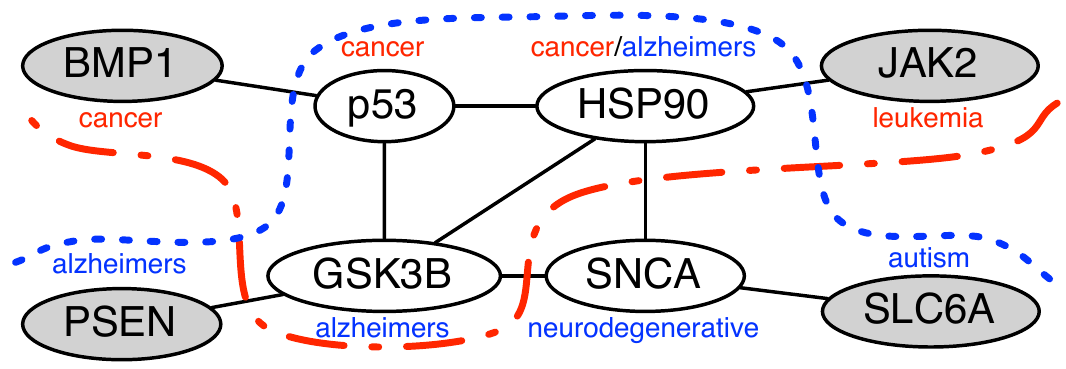}
\end{center}
\vspace{-3mm}
\caption{A minimum Wiener connector extracted from a PPI network and links genes associated with cancer and Alzheimer's disease.\label{fig:usecase6-bio}}
\vspace{-3mm}
\end{figure}

\spara{Protein-Protein-Interaction network.} \enlargethispage*{2\baselineskip}
Network analysis has established itself as a central component of computational and systems biology.
Barabasi et al \cite{barabasi2011network} drew attention to the great potential of \emph{``network medicine''} in the study of diseases.
This work highlighted the utility of identifying not only vertices with high betweenness centrality, but also those that act as links between diseases.
Finding such vertices may lead to the discovery of new protein-disease associations and a deeper understanding of the relationships between diseases~\cite{chang2013dynamic,chandrasekaran2013network,Santiago2014694,baryshnikova2013genetic}.

The minimum Wiener connector fits well in this setting, as it aims at finding few central vertices that connect a given set of query vertices.
%
%
As a proof of concept we use a human Protein-Protein-Interaction (PPI) network collected from BioGrid\footnote{\scriptsize\url{http://thebiogrid.org/}} with $15\,312$ vertices.  To demonstrate the utility of \wsq\ we require a ground truth about the relationship, so we select query proteins that have been the subject of previous biological study.
In Figure~\ref{fig:usecase6-bio} we report the minimum Wiener connector for a query set
shown in grey, and the solution connector-vertices in white.  For each query node we analyze the disease-association of its next-hop in the connector, and find that it is indicative of the studied association of the query node.  For example, we observe that the next-hop of \emph{BMP1} is \emph{p53} which is widely regarded as central in cancer; we then verify in the literature that in fact \emph{BMP1} has also been linked to cancer~\cite{vogelstein2010p53,thawani2010bone}.
Similar literature-verified examples are:
\squishlist
	\item  \emph{PSEN} is related to the other query nodes through \emph{GSK3B} -- uncovering its role in Alzheimer's disease.
	\item \emph{JAK2} is connected through \emph{HSP90} which has been studied for its potential therapeutic role in JAK-related diseases~\cite{bareng2007potential,fridman2012interplay}.
	\item \emph{SLC6A4} is suspected to play a role in Alzheimer disease, and is connected to \emph{SNCA}, a known factor in Alzheimer's.
\squishend
Further, the high connectivity of the inner nodes insinuates a close relationship between Cancer and Alzheimer's (e.g., as seen by the interaction between \emph{p53} and \emph{GSK3B}), which has in fact been a topic of interest and study~\cite{proctor2010gsk3,jacobs2012gsk,gao2012gsk3}.

This sample query is exemplary of the quality and potential of finding the minimum Wiener connector.
Identifying not only high betweenness and important nodes, but also those that act as links, gives potential for new directions of investigation for protein-disease and disease-disease relationships.
The connector also provides a concise summary of the relationships that is amenable to visualization.

\begin{figure}[t!]
  \vspace{-1mm}
\begin{center}
	\includegraphics[scale=0.45]{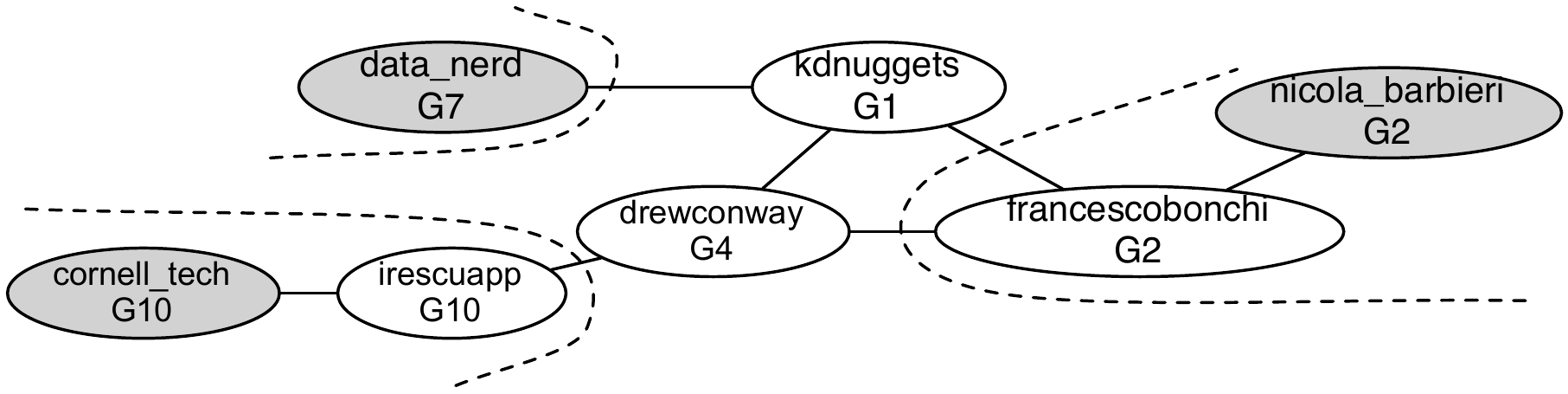}\\
	\includegraphics[scale=0.45]{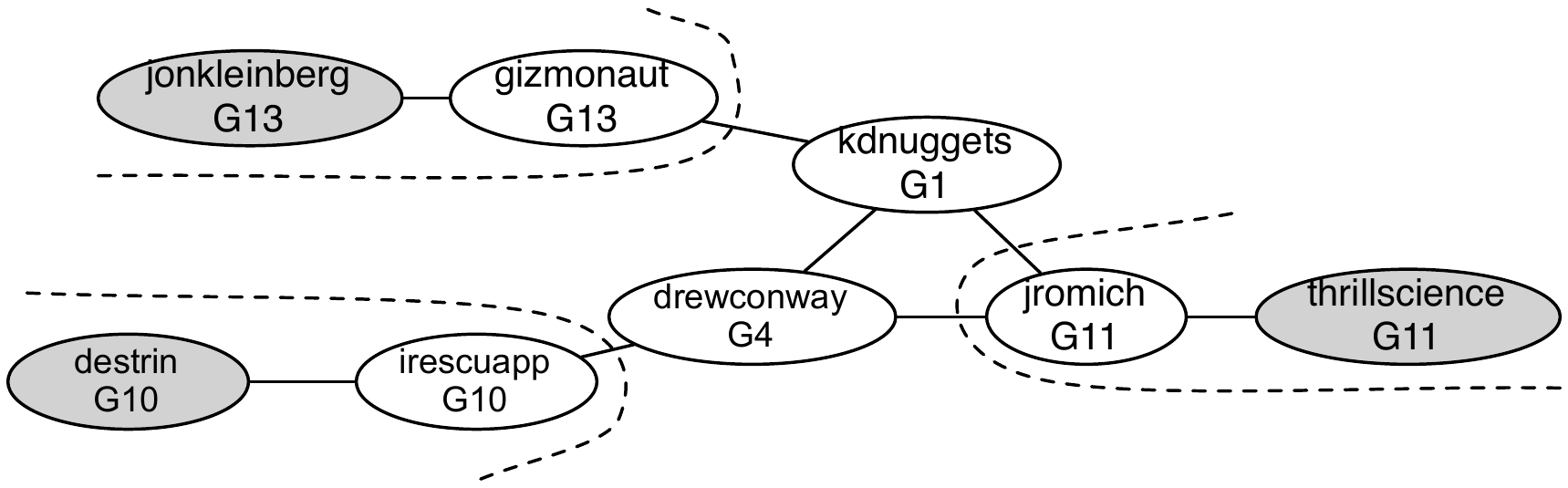}
\end{center}
  \vspace{-3mm}
\caption{Two minimum Wiener connectors extracted from the Twitter $\#kdd2014$ graph.\label{fig:usecase1-2}}
  \vspace{-3mm}
\end{figure}

\spara{Social network.}
The next case study is based on a graph created over Twitter users taking part in the ACM SIGKDD 2014 conference.
The graph contains 1\,141 Twitter users whose tweets over the three-day period contained the hashtag $\#kdd2014$, or who replied to or were mentioned in these tweets. There is an edge between two users for each reply or mention. \enlargethispage*{2\baselineskip}
The Clauset-Newman-Moore algorithm was used to cluster the graph into 10 communities.\footnote{\scriptsize\url{https://nodexlgraphgallery.org/Pages/Graph.aspx?graphID=26533}}

Figure~\ref{fig:usecase1-2} reports two minimum Wiener connectors extracted with query sets $Q$  (shown in gray)  consisting of vertices belonging to different communities.
The vertices chosen to be combined with $Q$ to produce the solution subgraph~$H$ are, in both cases, users that exhibit some influence or leadership.
In particular, we observe that in both examples,~$H$ contains the users $kdnuggets$	and $drewconway$, each of which have a very large set of Twitter followers
 ($23.1k$ and $10k$ respectively), and turn out to be the top mentioned and replied-to users in the whole $\#kdd2014$ dataset.
Table~\ref{tab:twitter-stats} contains more detailed information. In particular it shows that the other intermediate vertices included in the minimum Wiener connectors also exhibit high levels of activity
and are among the top-10 mentioned or replied-to users in their respective communities.

\begin{table}[t!]
\vspace{2mm}
\caption{Statistics on Tweeters in $\#kdd2014$ graph.}
\centering
\scriptsize
\tabcolsep=0.0cm
\begin{tabular}{lcl}
\toprule
UserId			&	Followers &	\multicolumn{1}{c}{Notes}					\\
\midrule
$kdnuggets$		&	23.1k	&	Top-1 mentioned in entire graph \& G1			\\
				& 			&	Top-1 betweenness in entire graph				\\
				&			&	Top-3 word in entire graph	 \& G1				\\
				&			&	Top-6 mentioned in G2-G8					\\
				&			&	Top-10 word in G4							\\
				&			&	Top-2 replied-to in G8						\\
$drewconway$		&	10.7k	&	Top-7 mentioned in entire graph \& G1, G4		\\
				&			&	Top-4 replied-to in entire graph \& G2, G3			\\
				&			&	Top-6 word in G4							\\
$gizmonaut$		&	304		&	Top-9 tweeter in G10							\\
$irescuapp$		&	204		&	Top-7 mentioned in G10						\\
$jromich$			&	165		&	Top-7 replied-to in G1						\\
$francescobonchi$	&	619		&	Top-7 mentioned in G2						\\
\bottomrule
\end{tabular}
\label{tab:twitter-stats}
\end{table}

%

\section{Conclusions}
\label{sec:conclusions}
In this paper we introduced the \ourprob\ problem: given a graph and a set of query vertices, find a subgraph that connects the query vertices while minimizing  the sum of pairwise shortest-path distances within that subgraph. In such simple and elegant formulation, the objective function favors small solutions built by adding important (central) vertices to connect the given query vertices.
Thanks to these features, the minimum Wiener connector lends itself naturally to applications in biological and social network analysis.

We showed that the problem is \NPhard, cannot admit any PTAS, and has an exact (yet impractical) algorithm that runs in polynomial time for the special case where the size of the query set is constant. Also, as a major contribution, we provided a constant-factor approximation algorithm that runs in time proportional (up to logarithmic factors) to the size of the input graph and the number of query vertices.

\pagebreak


\begin{thebibliography}{10}

\bibitem{Andersen2}
R.~Andersen, F.~R.~K. Chung, and K.~J. Lang.
\newblock Local graph partitioning using {PageRank} vectors.
\newblock In {\em FOCS 2006}.

\bibitem{Andersen1}
R.~Andersen and K.~J. Lang.
\newblock Communities from seed sets.
\newblock In {\em WWW 2006}.

\bibitem{Asur}
S.~Asur and S.~Parthasarathy.
\newblock A viewpoint-based approach for interaction graph analysis.
\newblock In {\em KDD 2009}.

\bibitem{bader06parallel-betweenness}
D.~Bader and K.~Madduri.
\newblock Parallel algorithms for evaluating centrality indices in real-world
  networks.
\newblock In {\em Int. Conf. on Parallel Processing}, pages 539--550, 2006.

\bibitem{barabasi2011network}
A.-L. Barab{\'a}si, N.~Gulbahce, and J.~Loscalzo.
\newblock Network medicine: a network-based approach to human disease.
\newblock {\em Nature Reviews Genetics}, 12(1):56--68, 2011.

\bibitem{bareng2007potential}
J.~Bareng, I.~Jilani, M.~Gorre, H.~Kantarjian, F.~J. Giles, A.~Hannah, and
  M.~Albitar.
\newblock A potential role for {HSP90} inhibitors in the treatment of {JAK2}
  mutant-positive diseases as demonstrated using quantitative flow cytometry.
\newblock {\em Leukemia \& lymphoma}, 48(11):2189--2195, 2007.

\bibitem{baryshnikova2013genetic}
A.~Baryshnikova, M.~Costanzo, C.~L. Myers, B.~Andrews, and C.~Boone.
\newblock Genetic interaction networks: toward an understanding of
  heritability.
\newblock {\em Annual review of genomics and human genetics}, 14:111--133,
  2013.

\bibitem{bavelas48}
A.~Bavelas.
\newblock A mathematical model of group structure.
\newblock {\em Human Organizations}, 7:16--30, 1948.

\bibitem{Burt92structuralholes}
R.~Burt.
\newblock {\em Structural Holes: The Social Structure of Competition}.
\newblock Harvard University Press, 1992.

\bibitem{best_steiner}
J.~Byrka, F.~Grandoni, T.~Rothvoss, and L.~Sanit\`{a}.
\newblock {S}teiner tree approximation via iterative randomized rounding.
\newblock {\em J. ACM}, 60(1):6:1--6:33, Feb. 2013.

\bibitem{CEla2011}
E.~\c{C}ela, N.~S. Schmuck, S.~Wimer, and G.~J. Woeginger.
\newblock The {W}iener maximum quadratic assignment problem.
\newblock {\em Discret. Optim.}, 8(3):411--416, 2011.

\bibitem{chandrasekaran2013network}
S.~Chandrasekaran and D.~Bonchev.
\newblock A network view on {P}arkinson's disease.
\newblock {\em Computational and structural biotechnology journal}, 7(8):1--18,
  2013.

\bibitem{chang2013dynamic}
X.~Chang, T.~Xu, Y.~Li, and K.~Wang.
\newblock Dynamic modular architecture of protein-protein interaction networks
  beyond the dichotomy of 'date' and 'party' hubs.
\newblock {\em Scientific reports}, 3, 2013.

\bibitem{OverlappingCSSIGMOD13}
W.~Cui, Y.~Xiao, H.~Wang, Y.~Lu, and W.~Wang.
\newblock Online search of overlapping communities.
\newblock In {\em SIGMOD}, pages 277--288, 2013.

\bibitem{SozioLocalSIGMOD14}
W.~Cui, Y.~Xiao, H.~Wang, and W.~Wang.
\newblock Local search of communities in large graphs.
\newblock In {\em SIGMOD}, pages 991--1002, 2014.

\bibitem{dynamic_hashing}
M.~Dietzfelbinger, A.~R. Karlin, K.~Mehlhorn, F.~{Meyer auf der Heide},
  H.~Rohnert, and R.~E. Tarjan.
\newblock Dynamic perfect hashing: Upper and lower bounds.
\newblock {\em {SIAM} J. Comput.}, 23(4):738--761, 1994.

\bibitem{DinurSafra05}
I.~Dinur and S.~Safra.
\newblock On the hardness of approximating minimum vertex cover.
\newblock {\em Annals of Mathematics}, 162(1):439--485, 2005.

\bibitem{Dobrynin2001}
A.~Dobrynin, R.~Entringer, and I.~Gutman.
\newblock {W}iener index of trees: Theory and applications.
\newblock {\em Acta Applicandae Mathematica}, 66(3):211--249, 2001.

\bibitem{connect}
C.~Faloutsos, K.~S. McCurley, and A.~Tomkins.
\newblock Fast discovery of connection subgraphs.
\newblock In {\em KDD}, 2004.

\bibitem{Fink2012}
J.~Fink, B.~Lu\v{z}ar, and R.~\v{S}krekovski.
\newblock Some remarks on inverse {W}iener index problem.
\newblock {\em Discrete Appl. Math.}, 160:1851--1858, 2012.

\bibitem{Fischermann2002}
M.~Fischermann, A.~Hoffmann, D.~Rautenbach, L.~Sz\'ekely, and L.~Volkmann.
\newblock {W}iener index versus maximum degree in trees.
\newblock {\em Discrete Appl. Math.}, 122(1--3):127--137, 2002.

\bibitem{Hwang1992}
D.~S.~R. Frank K.~Hwang and P.~Winter, editors.
\newblock {\em The {S}teiner Tree Problem}.
\newblock Annals of Discrete Mathematics. Elsevier, 1992.

\bibitem{fridman2012interplay}
J.~S. Fridman and N.~J. Sarlis.
\newblock The interplay between inhibition of {JAK2} and {HSP90}.
\newblock {\em JAK-STAT}, 1(2):77--79, 2012.

\bibitem{fujiwara12fast-topksearchRW}
Y.~Fujiwara, M.~Nakatsuji, M.~Onizuka, and M.~Kitsuregawa.
\newblock Fast and exact top-k search for random walk with restart.
\newblock {\em Proc. VLDB Endow.}, 5(5):442--453, 2012.

\bibitem{gao2012gsk3}
C.~Gao, C.~H{\"o}lscher, Y.~Liu, and L.~Li.
\newblock Gsk3: a key target for the development of novel treatments for type 2
  diabetes mellitus and {Alzheimer} disease.
\newblock {\em Reviews in the Neurosciences}, 23(1):1--11, 2012.

\bibitem{girvan02community}
M.~Girvan and M.~E.~J. Newman.
\newblock Community structure in social and biological networks.
\newblock {\em National Academy of Sciences of USA}, 99(12):7821--7826, June
  2002.

\bibitem{ellipsoid_lp}
M.~Gr\"otschel, L.~Lov\'asz, and A.~Schrijver.
\newblock The ellipsoid method and its consequences in combinatorial
  optimization.
\newblock {\em Combinatorica}, 1(2):169--197, 1981.

\bibitem{gurobi}
{Gurobi Optimization, Inc.}
\newblock Gurobi optimizer reference manual, 2015.

\bibitem{PPR2}
T.~H. Haveliwala.
\newblock Topic-sensitive pagerank.
\newblock In {\em WWW 2002}.

\bibitem{Hu1974}
T.~C. Hu.
\newblock Optimum communication spanning trees.
\newblock {\em SIAM J. Comput.}, 3(3):188--195, 1974.

\bibitem{KtrussSIGMOD14}
X.~Huang, H.~Cheng, L.~Qin, W.~Tian, and J.~X. Yu.
\newblock Querying k-truss community in large and dynamic graphs.
\newblock In {\em SIGMOD}, pages 1311--1322, 2014.

\bibitem{jacobs2012gsk}
K.~M. Jacobs, S.~R. Bhave, D.~J. Ferraro, J.~J. Jaboin, D.~E. Hallahan, and
  D.~Thotala.
\newblock {GSK-3}: A bifunctional role in cell death pathways.
\newblock {\em International journal of cell biology}, 2012, 2012.

\bibitem{PPR1}
G.~Jeh and J.~Widom.
\newblock Scaling personalized web search.
\newblock In {\em WWW 2003}.

\bibitem{kempe03}
D.~Kempe, J.~M. Kleinberg, and {\'E}.~Tardos.
\newblock Maximizing the spread of influence through a social network.
\newblock In {\em KDD}, 2003.

\bibitem{spanning_shortest}
S.~Khuller, B.~Raghavachari, and N.~E. Young.
\newblock Balancing minimum spanning trees and shortest-path trees.
\newblock {\em Algorithmica}, 14(4):305--321, 1995.

\bibitem{KleinRavi}
P.~Klein and R.~Ravi.
\newblock A nearly best-possible approximation algorithm for node-weighted
  {{S}teiner} trees.
\newblock {\em Journal of Algorithms}, 19(1):104--115, July 1995.

\bibitem{Kloumann}
I.~M. Kloumann and J.~M. Kleinberg.
\newblock Community membership identification from small seed sets.
\newblock In {\em KDD 2014}.

\bibitem{KorenTKDD07}
Y.~Koren, S.~C. North, and C.~Volinsky.
\newblock Measuring and extracting proximity graphs in networks.
\newblock {\em TKDD}, 1(3), 2007.

\bibitem{kossinets2006empirical}
G.~Kossinets and D.~J. Watts.
\newblock Empirical analysis of an evolving social network.
\newblock {\em Science}, 311(5757):88--90, 2006.

\bibitem{kourtellis12kpath}
N.~Kourtellis, T.~Alahakoon, R.~Simha, A.~Iamnitchi, and R.~Tripathi.
\newblock Identifying high betweenness centrality nodes in large social
  networks.
\newblock {\em Social Network Analysis and Mining}, 3:899--914, 2013.

\bibitem{Mehlhorn88}
K.~Mehlhorn.
\newblock A faster approximation algorithm for the {S}teiner problem in graphs.
\newblock {\em Inf. Proc. Letters}, 27(3):125--128, 1988.

\bibitem{Mohar1988}
B.~Mohar and T.~Pisanski.
\newblock How to compute the {W}iener index of a graph.
\newblock {\em J. Mathematical Chemistry}, 2(3):267--277, 1988.

\bibitem{cuckoo}
R.~Pagh and F.~F. Rodler.
\newblock Cuckoo hashing.
\newblock {\em J. Algorithms}, 51(2):122--144, 2004.

\bibitem{proctor2010gsk3}
C.~J. Proctor and D.~A. Gray.
\newblock {GSK3} and p53-is there a link in {Alzheimer}'s disease?
\newblock 2010.

\bibitem{riondato2014approxbetweenness}
M.~Riondato and E.~M. Kornaropoulos.
\newblock {Fast Approximation of Betweenness Centrality Through Sampling}.
\newblock In {\em WSDM}, 2014.

\bibitem{Rozenshtein2014}
P.~Rozenshtein, A.~Anagnostopoulos, A.~Gionis, and N.~Tatti.
\newblock Event detection in activity networks.
\newblock In {\em KDD}, pages 1176--1185, 2014.

\bibitem{Santiago2014694}
J.~A. Santiago and J.~A. Potashkin.
\newblock A network approach to clinical intervention in neurodegenerative
  diseases.
\newblock {\em Trends in Molecular Medicine}, 20(12):694 -- 703, 2014.

\bibitem{SozioKDD10}
M.~Sozio and A.~Gionis.
\newblock The community-search problem and how to plan a successful cocktail
  party.
\newblock In {\em KDD}, pages 939--948, 2010.

\bibitem{Spielman}
D.~A. Spielman and S.~Teng.
\newblock Nearly-linear time algorithms for graph partitioning, graph
  sparsification, and solving linear systems.
\newblock In {\em STOC 2004}.

\bibitem{Stefanovic2008}
D.~Stefanovic.
\newblock Maximizing {W}iener index of graphs with fixed maximum degree.
\newblock {\em MATCH Commun. Math. Comput. Chem.}, 60:71--83, 2008.

\bibitem{thawani2010bone}
J.~P. Thawani, A.~C. Wang, K.~D. Than, C.-Y. Lin, F.~La~Marca, and P.~Park.
\newblock Bone morphogenetic proteins and cancer: review of the literature.
\newblock {\em Neurosurgery}, 66(2):233--246, 2010.

\bibitem{distance_oracles}
M.~Thorup and U.~Zwick.
\newblock Approximate distance oracles.
\newblock {\em J. {ACM}}, 52(1):1--24, 2005.

\bibitem{CenterpieceKDD06}
H.~Tong and C.~Faloutsos.
\newblock Center-piece subgraphs: problem definition and fast solutions.
\newblock In {\em KDD}, pages 404--413, 2006.

\bibitem{tong06fast-RWR}
H.~Tong, C.~Faloutsos, and J.-Y. Pan.
\newblock Fast random walk with restart and its applications.
\newblock In {\em ICDM}, pages 613--622, 2006.

\bibitem{vogelstein2010p53}
B.~Vogelstein, S.~Sur, and C.~Prives.
\newblock p53: the most frequently altered gene in human cancers.
\newblock {\em Nature Education}, 3(9):6, 2010.

\bibitem{Wang2008}
H.~Wang.
\newblock The extremal values of the {W}iener index of a tree with given degree
  sequence.
\newblock {\em Discrete Appl. Math.}, 156(14):2647--2654, 2008.

\enlargethispage{\baselineskip}

\bibitem{wiener1947structural}
H.~{W}iener.
\newblock Structural determination of paraffin boiling points.
\newblock {\em J. Am. Chem. Soc}, 69(1):17--20, 1947.

\bibitem{YangL15}
J.~Yang and J.~Leskovec.
\newblock Defining and evaluating network communities based on ground-truth.
\newblock {\em Knowl. Inf. Syst.}, 42(1):181--213, 2015.

\bibitem{yu2014reversetopk}
A.~W. Yu, N.~Mamoulis, and H.~Su.
\newblock Reverse top-k search using random walk with restart.
\newblock {\em Proc. VLDB Endow.}, 7(5):401--412, 2014.

\bibitem{karate}
W.~Zachary.
\newblock An information flow model for conflict and fission in small groups.
\newblock {\em J. Anthropol. Res.}, 33(4):452--473, 1977.

\bibitem{Zhang2008}
X.-D. Zhang and Q.-Y. Xiang.
\newblock The {W}iener index of trees with given degree sequences.
\newblock {\em MATCH Commun. Math. Comput. Chem.}, 60:623--644, 2008.

\end{thebibliography}

\appendix
\section{Remaining proofs}\label{sec:more_proofs}

\subsection{Proof of Theorem~\ref{thm:exact_ub} (Section~\ref{sec:algorithms})}\label{sec:more_proofs_exact}

Recall that a \emph{homomorphism}  between two graphs $H$ and~$H'$ is a mapping $\phi: V(H) \to V(H')$ such that
$(u, v) \in E(H)$ implies $(\phi(u), \phi(v)) \in E(H')$.
\begin{lemma}\label{lem:homomorphism}
Let $\phi: H \to H$ be a surjective graph homomorphism. Then $\W(H') \le \W(H)$.
\end{lemma}
\begin{proof}
The existence of a homomorphism implies that for every path $p$ in $H$, there is a corresponding path in $H'$ (which may not be simple even if
        $p$ is). Therefore $d_{H'}(\phi(u), \phi(v)) \le d_H(u, v)$,
    and
\begin{align*}
\W(H) &=   \sum_{u, v \in V(H)} d_{H}(u, v) \\
      &\ge \sum_{u, v \in V(H)} d_{H'}(\phi(u), \phi(v)) \\
      &\ge \sum_{u', v' \in \phi(V(H))} d_{H'}(u', v') \\
      &= \sum_{u',v' \in V(H')} d_{H'}(u', v').
      \end{align*}

      The second inequality uses the fact that every pair $u', v'$ from the image of $V(H)$ under $\phi$ is counted at least once as $d_{H'}(\phi(u), \phi(v))$ for some $u, v \in V(H)$. 
The last equality is by the surjectivity of $\phi$.    
      \end{proof}

\begin{lemma}\label{lem:h_good}
Let $G$ be a graph, $H$ be a connected subgraph of $G$, $Q \subseteq V(H)$ and let 
$$A = Q \cup \{ v \in V(H) \mid \deg_H(v) > 2 \}.$$
We call $A$ the set of \emph{pivotal} vertices of $H$ with respect to $Q$.

Call a path $p$ between two vertices of $A$ \emph{basic} if the internal vertices of $p$ are outside $A$; say that an unordered pair of vertices $u, v
\in A$ is \emph{neighbouring} if $u, v \in V(H)$ and there is a basic path from $u$ to $v$.
Suppose we construct a graph $H'$ by including the vertices and edges of an arbitrary shortest path in $G$ between each pair of
neighbouring elements of $A$. Then $Q \subseteq V(H')$ and either $\W(H') = \W(H)$, or $H$ is not a minimum Wiener connector for $Q$.
\end{lemma}
\begin{proof}
It suffices to show that if $H$ is a minimum Wiener connector for $Q$, 
   then we can construct a surjective homomorphism $\phi$ from $H$ to $H'$ whose
restriction to $A$ is the identity. Indeed,  then clearly $Q \subseteq A \subseteq
V(H')$, and Lemma~\ref{lem:homomorphism} would imply $\W(H') \le \W(H)$.

For each neighbouring pair $(u, v)\in A \times A$, there is a unique basic path $p = p_0 \ldots p_t$ in $H$ between $u =
p_0$ and $v = p_t$ (otherwise removing some path yields a smaller connector).
The graph~$H'$ contains a path $q_0, \ldots, q_{t'}$ in $G$ (where $t' \le t$);
define $\phi(p_i) = q_{\min(i, t')}$ for $i \in \{0, \ldots, t'\}$. The map~$\phi$ is well-defined because the internal
    vertices of all these paths in $H$ are distinct, as their degree is 2. 
By construction,~$\phi$ is surjective on $V(H')$ and $\phi(u) = u$ for all $u \in A$. Moreover, the image of every edge in the unique basic path between a pair of neighbouring vertices
is an edge of $H'$. Since any edge of $H$ must belong to some basic path  (or else we could remove one of its endpoints while reducing the Wiener index of $H$), it
follows that $\phi$ is a homomorphism, as claimed.
\end{proof}


\begin{lemma}\label{lem:great}
Let $H$ be a connected graph and let $\calP = (\{s_i, t_i\})_{i \in [m]}$ denote a sequence of $m$ unordered pairs of distinct
vertices of~$H$.  Write $T = \bigcup_{i \in [m]} \{ s_i, t_i\}$ and call a sequence $p_1, \ldots, p_m$ of paths in $H$ \emph{valid} if for all $i \in [m]$, $p_i$ is a shortest path between $s_i$ and $t_i$.
There is a valid sequence $p_1, \ldots, p_m$ of paths in $H$
such that in the subgraph $H' = \bigcup_{i\in[m]} p_i$ of $H$ there are at most $m (m - 1)$ vertices with degree different from two, and $\W(H') \le
\W(H)$.
\end{lemma}
\begin{proof}
Since $H$ is connected, there is some valid sequence of paths; what we need to show is that there is one with the degree property.
We argue by induction on~$m$. When $m = 1$, all the internal vertices of any shortest pair between $s_1$ and $t_1$ have degree two, so we are done.
Suppose the theorem holds for $m-1$ vertex pairs; let $p_1, \ldots, p_{m-1}$ be the paths in the conclusion of the lemma, and take an arbitrary
shortest path $p_m$ from $s_m$ to $t_m$. We claim that we can replace the path $p_m$ with another path $p'_m$ such that, for each $p_i$ with $i < m$,
         at most two vertices have a different successor or predecessor in the two paths $p'_m$ and $p_i$. This is easy to see because if~$p_m$
         meets $p_j$, leaves it and intersects it for a second time, we can replace the part of $p_m$ between the two intersections by a subpath of $p_j$.

Consequently, when we add the edges in path $p'_m$ to the subgraph 
$\bigcup_{i\in[m-1]} p_i$, for any $i < m$ we increase the degree of at most two vertices that belong to $p_i$. Therefore the total number of vertices
of degree larger than two is at most $(m - 1)(m - 2) + 2(m - 1) = m (m - 1)$, as desired.

Finally, note that the above path-replacement procedure cannot increase the Wiener index as $H'$ is a subgraph of $H$ that maintains shortest-path
distances.
\end{proof}

\begin{lemma}\label{lem:few_pivotal}
Let $k = |Q|$. For any graph $G$, there is an optimal solution to \ourprob\ where at most $k^4$ vertices have degree in $H$ larger than two.
\end{lemma}
Note that such $H$ is not necessarily an induced subgraph.
\begin{proof}
Let $H$ be an optimal solution.
Consider a sequence $\calP$ containing all $m = \binom{k}{2}$ pairs of distinct query nodes.
By Lemma~\ref{lem:great}, there is a sequence $p_1, \ldots, p_m$ of shortest paths in $H$ (one for each pair of query nodes) such that in the graph $H'$
formed by the union of all these paths, there are most $m (m - 1)$ vertices with  degree different from two. Clearly $H'$ is a connector for $Q$
(since it contains paths linking each pair of query nodes). Since $\W(H') \le \W(H)$ and~$H$ is
optimal, it follows that $\W(H') = \W(H)$. Also, there cannot be any non-query vertices of degree 1, otherwise we could remove them from $H'$ and
still obtain a connector of $Q$ with smaller Wiener index. So the total number of vertices of degree larger than 2 in $H'$ is at most $k + m (m - 1) \le k^4$.
\end{proof}

\spara{Proof of Theorem~\ref{thm:exact_ub}}.
Let $k = |Q|$. We can loop over all possible $\binom{n}{k^4}$ vertex subsets of size $k^4$; by Lemma~\ref{lem:few_pivotal} one of them will be the set $X$ of vertices of
degree 2 in the optimal solution $H^*$. Then $X \cup Q$ is the set of pivotal vertices of~$H^*$ with respect to $Q$; 
we can construct in polynomial time a graph $H'$ as in Lemma~\ref{lem:h_good}, and we will have $\W(H') \le \W(H^*)$,
   hence $\W(H^*) = W(H')$. 

Overall, the algorithm runs in time $n^{\poly(k)}$.\qed

\subsection{Proof of Lemma~\ref{lem:one_source} (Section~\ref{sec:approxalg})}\label{sec:more_proofs_one_source}
Let $r^* = \operatorname{argmin}_{r} \sum_v d_H(v, r)$.
Observe that
\begin{align*}
|V(H)| \cdot \sum_v d_H(v, r^*)
&= \sum_w \big( \sum_v d_H(v, r^*) \big) \\
&\le \sum_w \big (\sum_v d_H(v, w) \big) 
&= 2\,\W(H), 
\end{align*}
\vspace{-0.4cm}
\begin{align*}
\text{and } 2\, \W(H) &\le \sum_w \big (\sum_v [ d_H(v, r^*) + d_H(r^*, w) ] \big) \\
&= \sum_{v,w} d_H(v, r^*) + \sum_{v,w} d_H(r^*, w)  \\
&= 2 \sum_{v,w} d_H(v, r^*) 
= 2\,|V(H)| \cdot \sum_{v} d_H(v, r^*),
\end{align*}
where we used the choice of $r^*$ in the first inequality and the triangle inequality for $d_H$ in the second inequality.\qed

\subsection{Proof of Lemma~\ref{lem:add_shortest_paths} (Section~\ref{sec:approxalg})}\label{sec:proof_tree}
Let $T_S$ be the shortest-path tree from $r$ to the elements of~$T$, determined by an array of distances $d_S[]$ and an array of parent links $p_S[]$.
Consider the algorithm below, which traverses $T$ and performs a series of edge relaxations that add additional vertices
from $T_S$ in order to decrease distances to the root~$r$. 
The important invariant maintained is that the edges $\{(p[v], v)\} \mid d[v] \neq \infty\}$ form a
subtree of $T \cup T_S$, with $d[v]$ an upper bound on the distance between the root and $v$ in the tree.

\begin{algorithm}[h]
\caption{\textsf{AdjustDistances}}
\small
\begin{algorithmic}[1]
\Require A graph $G=(V,E)$; a subtree $T$; a root node $v \in V(T)$; and a BFS tree from $r$ with parent array $p_S[]$ and distance array $d_S[]$.
\Ensure A tree.\vspace{2mm}

\State Construct hash tables $d[], p[]$, with default values $p[v] = \textbf{nil}$ and $d[v] = \infty$ for all $v \in V(G)$.
\State $d[r] \gets 0$.
\State $\textsf{dfs}(r)$.
\State \Return the tree $T'=\{(v, p[v]) \mid v \in V(G) \wedge p[v] \neq \textbf{nil}\}$.
\end{algorithmic}
\end{algorithm}

\begin{algorithm}[h]
\caption{\textsf{dfs}}
\small
\begin{algorithmic}[1]
\Require A vertex $u$.
\If {$d[u] > (1 + \sqrt{2}) d_S[u]$}
    \State $\textsf{AddPath}(u)$
\EndIf    

\For{each child $v$ of $u$ in $T$}
    \State $\textsf{relax}(u, v)$
    \State $\textsf{dfs}(v)$
    \State $\textsf{relax}(v, u)$
\EndFor
\end{algorithmic}
\end{algorithm}

\begin{algorithm}[h]
\caption{\textsf{AddPath}}
\small
\begin{algorithmic}[1]
\Require A vertex $u$.
\State $v \gets u$
\While {$d[v] > d_S[v]$}
    \State $\textsf{relax}(p_S[v], v)$
    \State $v \gets p[v]$
\EndWhile
\end{algorithmic}
\end{algorithm}

\begin{algorithm}[h]
\caption{\textsf{Relax}}
\small
\begin{algorithmic}[1]
\Require two adjacent vertices $u, v$.
\If {$d[v] > d[u] + 1$}
    \State $d[v] \gets d[u] + 1$
    \State $p[v] \gets u$
\EndIf    
\end{algorithmic}
\end{algorithm}
    
We need to show that \textsf{AdjustDistances} runs in time $O(|V(T)|)$ and returns a tree $T'$ satisfying the following:
\begin{enumerate}[(a)]
    \item $V(T') \supseteq V(T)$;
    \item $|V(T')| \le (1 + \sqrt{2}) |V(T)|$;
    \item for all $v \in V(T')$, $d_{T'}(r, v) \le (1 + \sqrt{2})\,d_G(r, v)$.
    \item $\sum_{v \in V(T')} d_{G}(r, v) \le \sqrt{2}\, \sum_{v \in V(T)} d_G(r, v)$.
\end{enumerate}

Property a) holds because every time we insert a new vertex $u$ in the tree (that is, $p[u]$ becomes $\neq \textsf{nil}$), it is never removed again; 
and the call $\textsf{dfs}(r)$ visits
all vertices of $T$. Property c) holds because  
$d[u] = d_S[u]$ for
all $u \in V(T') \setminus V(T)$, and
whenever $d[u] > (1 + \sqrt{2}) d_S[u]$
for some $u \in V(T)$, we add a path to achieve $d[u] = d_S[u]$.

Next we analyze the running time. For $d[]$ and $p[]$ we use a resizable hash table with constant expected amortized
update/lookup time~\cite{dynamic_hashing,cuckoo}. When an element $v$ is not in the table, we insert $p[v] \gets
\textbf{nil}$ and $d[v] \gets \infty$.  This way lines 1-2 of \textsf{AdjustDistances} take time $O(1)$.
We also keep track of which elements $v \in V(G)$ have been assigned values in the table, so line 4 takes time $O(|V(T')|)$ rather than $O(|V(G)|)$. 
The running time of $\textsf{dfs}(r)$ (excluding line 1, which is run $O(|V(T)|)$ times) is proportional to the number of calls to
$\textsf{relax}$ made by $\textsf{AddPath}$ and the recursive calls to $\textsf{dfs}$. 
The number of relaxations is $O(|V(T')|)$ because every edge of $T$ or $T_S$ is relaxed at most twice by \textsf{dfs} and at most once by \textsf{AddPath}.
Therefore the running time of $\textsf{dfs}(r)$ is $O(|V(T')|)$, which is also $O(|V(T)|)$ assuming property b).

Now we show property b). As the algorithm executes, define a potential function $\Phi$ to be the distance estimate of the current vertex (for ease of notation we omit the
        dependence of $\Phi$ on the current time). When a shortest path of length $\ell = d_S[u]$ to
the current vertex $u$ is added by $\textsf{AddPath}(u)$, $\phi = d[u] > \alpha \ell$, where $\alpha = 1 + \sqrt{2}$. Adding the path lowers $d[u]$ to
$\ell$, decreasing $\phi$ by
at least $(\alpha - 1) \ell$. Hence the total length of the added paths is bounded by the sum of the decrements to $\phi$ during the course of the
algorithm, divided by $\alpha - 1$.
Since $\phi$ is initially 0 and always nonnegative, the sum of the decreases is at most the sum of the increments. The potential $\Phi$ increases only when the
current vertex changes from some vertex $u$ to a vertex $v$ after the edge $(u, v)$ was relaxed, which ensures that $d[v] \le d[u] + 1$ and that $\Phi$
increases by at most~1. Since each edge is traversed twice, the total of the increases to $\Phi$ during the course of the algorithm is bounded by
twice the number of edges in~$T$.
This establishes that the total length of the added paths is bounded by $2/(\alpha - 1) = \sqrt{2}$ times the total number of edges of $T$.
Thus, $|V(T')| \le (1 + \sqrt{2}) |V(T)|$, showing b).

Only property d) remains to be shown. Define analogously a potential $\Psi$ to be $\binom{d[v]+1}{2}$ when the current vertex is $v
\in V(T)$. Adding a shortest path of length $\ell$ when $d[v] > (1 + \sqrt{2}) \ell$ lowers $\Psi$ by at least $\ell^2 (1 + \sqrt 2)$.
The vertices added increase the sum of distances from $r$ by at most $\binom{\ell-1}{2} \le \ell^2 / 2$. Hence the total increase in sum of
distances is bounded by the sum of the decrements to $\Psi$, divided by $2 (1+\sqrt{2})$.
The sum of the deceases is at most the sum of the increases. The potential~$\Psi$ increases by at most $d_G(v)$ when relaxing an
edge $(u, v)$. Since each edge is traversed twice, the total of the increases to $\Psi$ is bounded by
$\frac{2}{2(1 + \sqrt 2)} \sum_{v \in V(T)} d_G(v)$. Hence
$\sum_{v \in V(T') \setminus V(T)} d_G(v) \le \frac{1}{1 + \sqrt2} \sum_{v \in V(T)} d_G(v),$
which implies c).\qed

\subsection{Proof of Lemma~\ref{coro:alpha_approx} (Section~\ref{sec:approxalg})} \label{sec:more_proofs_alpha}
We need the following lemma.
\begin{lemma}\label{lem:sum_product}
Let $x_0, y_0 \in \reals^+$ and $\lambda = \sqrt{\frac{y_0}{x_0}}$.
Then, for all $x, y \in \reals^+$ it holds that
$$
 \frac{x y}{x_0 y_0} \le \left( \frac{x \lambda + \frac{y}{\lambda}}{x_0 \lambda + \frac{y_0}{\lambda}} \right)^2.
 $$
\end{lemma}
\begin{proof}
Our choice of $\lambda$ implies that $4x_0y_0 = \left ( x_0 \lambda + \frac{y_0}{\lambda} \right )^2$. 
Recall that, by the AM--GM inequality, $\sqrt{ab} \leq \frac{a+b}{2} \Rightarrow 4ab \leq (a + b)^2$ for all $a,b \in \reals^+$.
Hence
$$
\frac{xy}{x_0y_0} = \frac{4xy}{4x_0y_0} = \frac{4\left( x\lambda \right )\left ( \frac{y}{\lambda} \right )}{\left ( x_0\lambda + \frac{y_0}{\lambda}
        \right )^2} \leq \left ( \frac{x\lambda + \frac{y}{\lambda}}{x_0\lambda + \frac{y_0}{\lambda}} \right )^2.\qed
$$
\end{proof}

Now we are ready to show Lemma~\ref{coro:alpha_approx}. Let $\Wonesub^*$ denote the optimal solution to Problem~\ref{prob:weakrootedwiener} and
set 
$
\lambda = \sqrt{\frac{\sum_{u \in \Wonesub^*} d_G(r, u)}{|\Wonesub^*|}}. 
$
It is easy to see that $\lambda \in [1/\sqrt{2}, \sqrt{|V(G)|}]$ as all distances $d_G(r, u)$ but one (i.e., $d_G(r, u)$ which is equal to 0) are in the range $[1, |V(G)|]$ and $|A^*| \geq 2$.
Let $\Wtwosub^*$ (resp., $\Wtwosub$) denote the optimal solution (resp., an $\alpha$-approximate solution) to Problem~\ref{prob:weird}.
By our choice of $\lambda$, Lemma~\ref{lem:sum_product} implies that
\begin{eqnarray*}
\lefteqn{\frac{\Wonew(\Wtwosub,r)}{\Wonew(\Wonesubw^*,r)} = \frac{|\Wtwosub|\sum_{u \in \Wtwosub} d_G(r,u)}{|\Wonesubw^*|\sum_{u \in \Wonesubw^*} d_G(r,u)}}\\
& \leq & \left ( \frac{|\Wtwosub|\lambda + \frac{1}{\lambda}\sum_{u \in \Wtwosub} d_G(r,u)}{|\Wonesubw^*|\lambda + \frac{1}{\lambda}\sum_{u \in \Wonesubw^*} d_G(r,u)} \right )^2\\
&  = & \left ( \frac{\Wtwo(\Wtwosub, r)}{\Wtwo(\Wonesubw^*, r)}  \right )^2 \ \leq \ \left ( \frac{\Wtwo(\Wtwosub, r)}{\Wtwo(\Wtwosub^*, r)}  \right
        )^2,
\end{eqnarray*}
where the last inequality follows from $\Wtwo(\Wonesubw^*, r) \ge \Wtwo(\Wtwosub^*, r)$ (due to the optimality of $B^*$ for Problem~\ref{prob:weird}).
To complete the proof, note that $\frac{\Wtwo(\Wtwosub, r)}{\Wtwo(\Wtwosub^*, r)} \leq \alpha$ holds by hypothesis.\qed

\subsection{Proof of Theorem~\ref{ip_model} (Section~\ref{sec:lowerbounds})}\label{sec:more_proofs_ip} We have to show that an optimal solution to~\eqref{ip1} gives an optimal solution to our problem, and viceversa.

It is clear that any solution $S \subseteq V(G)$ to \ourprob\ yields a feasible solution to~\eqref{ip1}: set
$y_u = 1$ and $p_{uv} = 1$ iff $u, v \in S$, and for each $s, t \in S$, pick a shortest path $z_0 = s, z_1, \ldots, z_t
= t$ from $s$ to $t$ in $S$ and
set $f_{z_i,z_{i+1}}^{st} = 1$; set all other variables to $0$.
This satisfies all constraints and the objective function coincides with $\W(S)$.

Conversely, consider an optimal solution to~\eqref{ip1} and let $S = \{ u \in V \mid y_u = 1 \}$; note that $S \supseteq Q$. We show that
the objective function is at least $\W(G[S])$. For any $s, t \in S$, $p_{st} \ge y_s+y_t-1 \ge 1$. The constraints now imply that
we can route $p_{st} \ge 1$ units of flow from $s$ to $t$, where the capacity of directed edge $(u, v)$ is at
most~$y_u$. Note that once~$y_u$ and~$p_{st}$ have been fixed, all remaining constraints involve only flow variables and constants, and only flow variables with the same pair $s, t$. Therefore, the only
way to minimize the objective function is to find, for each $s, t$, a min-cost flow of value $p_{st}$ (where the cost and capacity of every edge is one), i.e., a
shortest-path from $s$ to~$t$ in $S$. Thus there is such a path~$p$,
and the sum of $f^{st}_{uv}$ for the directed edges $(u, v)$ of $p$ is at least $d_S(s, t)$. As this holds for all $s, t$, the result follows.\qed


\end{document}